\documentclass[11pt]{article}

\usepackage[a4paper,margin=1in]{geometry}
\usepackage{amsmath,amssymb,amsthm,mathtools}
\usepackage{titlesec}
\usepackage{bm}
\usepackage{booktabs}
\usepackage{graphicx}
\usepackage{microtype}
\usepackage{xcolor}
\usepackage{hyperref}
\usepackage{tikz}
\usetikzlibrary{arrows.meta,calc,positioning}

\hypersetup{
  colorlinks=true,
  linkcolor=blue!55!black,
  citecolor=blue!55!black,
  urlcolor=blue!55!black
}

\newtheorem{proposition}{Proposition}[section]
\newtheorem{corollary}{Corollary}[section]
\newtheorem{remark}{Remark}[section]

\newcommand{\R}{\mathbb{R}}
\newcommand{\E}{\mathbb{E}}
\newcommand{\Var}{\operatorname{Var}}
\newcommand{\Cov}{\operatorname{Cov}}
\newcommand{\op}{O_{\mathrm p}}
\newcommand{\ip}[2]{\left\langle #1,#2\right\rangle}
\newcommand{\thetahat}{\widehat{\bm\theta}}
\newcommand{\thetao}{\bm\theta_0}
\newcommand{\dtheta}{\bm\delta}

\title{\textbf{The Greedy Bump Bias: Local Profiling Geometry and the Look-Elsewhere Effect}}

\author{
Tommaso Dorigo\\[0.5em]
\small INFN, Sezione di Padova, Padova, Italy;\\
\small Department of Computer Science, Electrical and Space Engineering,\\
\small Luleå University of Technology, Luleå, Sweden
}
\date{\today}

\begin{document}

\maketitle

\begin{abstract}

When fitting a localized signal whose position or shape is not known in
advance, one typically allows these parameters to vary together with
the signal amplitude and chooses the values that maximize the
likelihood.  This freedom has two related statistical consequences.
If a genuine signal is present, its fitted amplitude will be affected by a positive bias; and if no signal is present, the same freedom increases the chance of finding an unusually signal-like background fluctuation, giving rise to the multiple-testing problem known in particle physics as the look-elsewhere effect.  We show that these two
effects can be understood as consequences of the same local geometry
of the family of signal templates.

We study this connection in a Gaussian matched-filter model, where a
smooth $D$-dimensional family of normalized templates describes the
unknown signal location or shape.  For a strong signal of fixed-template signal-to-noise ratio $Q$, profiling over the $D$ coordinates gives
\[
  \Delta q
  \equiv
  q_{\rm prof}-q_{\rm fixed}
  \xrightarrow{d}\chi^2_D,
  \qquad
  \E_Q[\widehat Q-Q]
  =
  \frac{D}{2Q}+o(Q^{-1}).
\]
Under the background-only hypothesis, the same $D$ local directions
make the null trials factor---the ratio of the global to the
fixed-template tail probability---grow at high threshold as
\[
  {\rm TF}_0(u)
  =
  K_{\mathcal M}u^D[1+o(1)].
\]
In the normalized matched-filter problem, the curvature of a genuine
signal peak and the fluctuations that determine the curvature of a
high background peak are governed by the same template metric.  This
allows us to derive the explicit asymptotic relation
\[
  \E_Q[\widehat Q-Q]
  =
  \frac12
  \left.
  \frac{d}{du}\log {\rm TF}_0(u)
  \right|_{u=Q}
  +o(Q^{-1}),
\]
with an analogous relation for the mean profile-likelihood gain.

We then follow the problem away from the strong-signal and
high-threshold limits.  Separating the signal-associated maximum from
the best competing maximum gives an exact decomposition of the global
bias into a local profiling contribution and a contribution from
remote-peak competition.  At finite threshold, background maxima are
not exactly a Poisson process: factorial cumulants quantify the
resulting corrections to the probability of finding no competing
peak.  In a one-dimensional Gaussian example, the second factorial
cumulant accounts for most of this correction, while the third brings
the prediction into close agreement with simulation.  A two-point
Kac--Rice calculation reproduces the second cumulant and reveals a
quartic short-distance suppression of nearby maxima.  The resulting
picture separates the roles of local dimension, model-dependent
curvature, and global extremal competition within a common framework.

\end{abstract}
\clearpage

\section{Introduction}
\label{sec:intro}

Searches for localized signals frequently involve model parameters that are completely unknown before a signal is found and measured.  The canonical example in high-energy physics (HEP) is a resonance search: the general procedure there is to fit for the signal amplitude while also allowing the resonance mass to vary; usually any additional shape parameters, such as the resonance width, are profiled in the fit.  
Closely related scan-and-optimize problems arise in source searches and
matched filtering~\cite{SiegmundWorsley1995,Zubeldia2021,Whitehorn2016},
in searches for spectral components of unknown
frequency~\cite{Davies1987}, and in change-point and scan-statistic
problems~\cite{SiegmundYakir2000}.


The statistical consequences of the profiling operation over additional parameters are usually discussed for the \emph{testing} stage.  Under the background-only hypothesis, the location of a putative signal is not identifiable, so the significance of a local excess must be corrected for the possibility that a similar fluctuation could have appeared elsewhere.  In HEP this is widely known as the look-elsewhere effect (LEE).  Building on Davies’ original Gaussian-process treatment and its later extension to chi-square processes
\cite{Davies1977,Davies1987}, Gross and Vitells gave an efficient
high-threshold treatment of the one-dimensional resonance-search problem and showed, in particular, that the trials factor grows asymptotically linearly with the local significance \cite{GrossVitells2010}.  Vitells and Gross subsequently generalized the random-field treatment to multidimensional searches \cite{VitellsGross2011}.

There is, however, a closely related \emph{estimation} effect.
In an unpublished internal CDF study carried out in 2003 and later
described in a blog post~\cite{Dorigo2009}, we observed that the fitted
yield of a Gaussian signal was biased upward when its position was
allowed to float.  The size of the effect decreased with increasing
signal-to-noise ratio.  At the time, unaware of the existing
literature reviewed in Sec.~\ref{sec:related}, we referred to this effect as the
\emph{greedy bump bias}.

The underlying intuition is simple.  Even when the signal model is
correct, its position or shape is estimated from the same noisy data as
its amplitude.  Allowing these parameters to vary lets the fitted
template adapt to the particular noise realization: it may move toward
a positive fluctuation or away from a negative one.  The displacement
itself has no preferred direction, but it is chosen so as to improve
the fit.  Consequently the corresponding change in the optimized peak
height is positive on average.  As shown below, locally this mechanism
has a precise form: a zero-mean noise gradient determines the fitted
displacement, and maximization converts it into a positive quadratic
gain.  The original CDF study established the effect numerically but
did not give an analytic account of this mechanism.

The purpose of the present work is to investigate whether the greedy
bump bias and the local geometry underlying the LEE are two
manifestations of the same profiling operation.  Our starting point is
deliberately simple: a Gaussian-noise template model in which the signal
amplitude is a radial coordinate and the unknown location or shape
parameters span a smooth template manifold.  In the whitened
data space, the possible signal mean vectors therefore form a cone over
the normalized template manifold.  Away from the cone apex, where a
signal is present, the model is locally regular and can be analyzed
through its tangent space.  At the apex, corresponding to the null
hypothesis, all template coordinates collapse to the same mean vector
and become unidentifiable.


Our first goal is local.  We ask how profiling over $D$ smooth signal coordinates changes the fitted amplitude and the likelihood gain when the signal is strong enough that the optimizer remains near the true
template.  We denote by $Q$ the true signal amplitude in units of the fixed-template standard deviation; in the Gaussian model used below, $Q$ is therefore the fixed-template signal-to-noise ratio.

Recasting known alternative-side peak asymptotics in the template
geometry used here, profiling exposes $D$ Gaussian tangent directions.
If
\[
  q \equiv 2\log\frac{L(\widehat{\text{signal model}})}
                         {L(\text{background-only model})}
\]
denotes the usual twice-log-likelihood-ratio statistic, then the gain
obtained by allowing the signal location or shape to vary,
\[
  \Delta q
  \equiv
  q_{\rm prof}-q_{\rm fixed},
\]
satisfies
\[
  \Delta q \;\xrightarrow{d}\; \chi^2_D,
\]
while the fitted signal amplitude obeys
\[
  \E[\widehat Q-Q]
  =
  \frac{D}{2Q}+o(Q^{-1}).
\]
Thus, even for a correctly specified signal, estimating its location or
shape from the same noisy data produces a positive fitted-amplitude
bias.  For $D=1$ this reproduces the basic greedy-bump effect; for
$D=2$ the leading bias doubles and the profile gain tends to a
$\chi^2_2$ distribution.

The broader question is how this local correspondence connects to the global search problem.  A strong signal anchors the optimizer near a distinguished template point, so the leading bias depends on local dimension but not on the total search volume.  As the signal weakens, however, a noise maximum elsewhere in the search region can overtake
the signal-associated peak and the global search geometry re-enters.
We therefore distinguish, within each realization, the local maximum associated with the signal from the highest competing maximum.  Their competition gives an exact decomposition of the global bias and makes the local-to-global crossover directly measurable.

\begin{figure}[p]
\centering
\begin{tikzpicture}[
    x=1cm,y=1cm,
    >=Latex,
    every node/.style={font=\small},
    paneltitle/.style={font=\bfseries},
    axis/.style={thick},
    vec/.style={-Latex,thick},
    guide/.style={dashed,gray!70},
    mylabel/.style={fill=white, inner sep=1.5pt}
]

\begin{scope}[shift={(0,8.1)}]

\node[paneltitle, anchor=west] at (0,6.35)
    {(a) Signal cone and local tangent geometry};

\coordinate (Apex)  at (5.7,0.15);
\coordinate (BaseC) at (5.7,4.65);

\draw[fill=blue!8, draw=blue!50!black] (BaseC) ellipse (2.0 and 0.48);

\filldraw[fill=blue!12, draw=blue!50!black]
    ($(BaseC)+(-2.0,0)$) -- (Apex) -- ($(BaseC)+(2.0,0)$)
    arc[start angle=0,end angle=180,x radius=2.0,y radius=0.48]
    -- cycle;

\draw[draw=blue!50!black]
    ($(BaseC)+(2.0,0)$) arc[start angle=0,end angle=-180,x radius=2.0,y radius=0.48];

\node[align=center] at (5.7,5.55)
    {$\mathcal C=\{\mu\,s(\theta):\mu\ge 0,\ \theta\in\mathcal M\}$};

\fill (Apex) circle (1.2pt);
\node[anchor=east] at ($(Apex)+(-0.12,0.02)$) {$H_0$};
\node[anchor=west] at ($(Apex)+(0.12,-0.35)$) {$\mu=0$};

\coordinate (P) at (6.65,3.15);

\coordinate (T1) at ($(P)+(-0.30,0.22)$);
\coordinate (T2) at ($(P)+(0.62,0.10)$);
\coordinate (T3) at ($(P)+(0.92,-0.25)$);
\coordinate (T4) at ($(P)+(0.00,-0.14)$);

\filldraw[fill=orange!22, draw=orange!70!black]
    (T1)--(T2)--(T3)--(T4)--cycle;

\fill[red!80!black] (P) circle (1.6pt);

\draw[vec, red!80!black] (Apex) -- (P);
\node[align=center] at (5.65,2.55) {amplitude\\ $\mu$};

\node[anchor=north] at ($(P)+(0.,+0.85)$) {$Q\,s(\theta_0)$};

\coordinate (MidTopRight) at ($(T2)!0.45!(T3)$);
\coordinate (MidBottomRight) at ($(T3)!0.48!(T4)$);

\draw[vec, orange!70!black] (MidTopRight) -- ++(1.20,-0.1)
    node[above right] {$\partial_{\theta_1}s$};

\draw[vec, orange!70!black] (MidBottomRight) -- ++(0.62,-0.95)
    node[below right, align=left] {$\partial_{\theta_2}s,\ldots,$\\[-2pt]$\partial_{\theta_D}s$};

\node[anchor=west, align=left] at (8.75,4.10)
    {local tangent space};

\node[align=center, text width=3.5cm] at (1.65,3.15)
    {for $\mu>0$: local regular model,\\ profiling explores $D$ tangent directions};

\node[align=center, text width=3.3cm] at (1.65,1.05)
    {at the apex: $\theta$ is\\ unidentifiable under $H_0$};

\end{scope}

\begin{scope}[shift={(0,0)}]

\node[paneltitle, anchor=west] at (0,5.6)
    {(b) Local maximum versus remote competitors};

\draw[axis, -Latex] (0.8,1.15) -- (11.35,1.15) node[below right] {$\theta$};
\draw[axis, -Latex] (0.8,1.15) -- (0.8,4.8) node[left] {$Z(\theta)$};

\fill[green!12] (4.15,1.15) rectangle (6.15,4.35);

\draw[thick, blue!65!black]
plot[smooth] coordinates {
    (1.10,1.50)
    (1.60,2.20)
    (2.00,1.55)
    (2.55,1.85)
    (3.20,2.00)
    (4.25,2.10)
    (4.70,2.70)
    (5.05,3.60)
    (5.20,3.88)
    (5.35,3.65)
    (5.65,2.45)
    (6.25,1.65)
    (7.15,1.45)
    (8.10,1.95)
    (8.60,2.95)
    (8.85,3.18)
    (9.10,2.88)
    (9.55,1.92)
    (10.40,1.72)
    (11.00,1.82)
};

\coordinate (ThetaZero) at (5.20,1.15);
\coordinate (MS) at (5.20,3.88);
\coordinate (MR) at (8.85,3.18);

\draw[guide] (ThetaZero) -- (MS);
\draw[guide] (8.85,1.15) -- (MR);
\draw[guide] (4.15,1.15) -- (4.15,4.0);
\draw[guide] (6.15,1.15) -- (6.15,4.0);

\node[below] at (4.15,1.15) {$\theta_0-A$};
\node[below] at (5.20,1.15) {$\theta_0$};
\node[below] at (6.15,1.15) {$\theta_0+A$};

\fill[red!80!black] (MS) circle (1.6pt);
\fill[purple!80!black] (MR) circle (1.6pt);

\node[mylabel, above left] at (MS) {$M_S$};
\node[mylabel, above right] at (MR) {$M_R$};

\node[mylabel] at (9.70,4.25) {$M=\max(M_S,M_R)$};

\node[align=center] at (2.15,0.35) {remote\\search region};
\node[align=center] at (5.15,0.35) {local\\signal region};
\node[align=center] at (10.00,0.35) {remote\\search region};

\end{scope}

\end{tikzpicture}
\caption{
Conceptual geometry of the greedy-bump problem and its connection to the look-elsewhere effect.
\textbf{(a)} The closure of the signal family forms a cone
$\mathcal C=\{\mu s(\theta): \mu\ge 0,\ \theta\in\mathcal M\}$.
At a nonzero signal point $Q\,s(\theta_0)$, where \(Q>0\) denotes the true signal amplitude, profiling over the nuisance coordinates explores the local tangent space and produces the strong-signal bias and profile gain.
At the apex $\mu=0$, corresponding to the null hypothesis, the template coordinates become unidentifiable.
\textbf{(b)} In a scan statistic or matched-filter field, the global maximum can be decomposed into the signal-associated local maximum $M_S$ and the best remote competitor $M_R$.
At large $Q$, $M_S$ dominates and the bias is local; at smaller $Q$, remote maxima can overtake it, reintroducing the global search-volume dependence.
}
\label{fig:geometry_bridge}
\end{figure}
A second issue appears once remote extrema matter.  Standard
high-excursion treatments exploit the asymptotically Poisson character
of sufficiently rare peaks.  At finite threshold, however, maxima of a
smooth field remain correlated.  Their mean number is therefore not
enough to determine the probability that no competing maximum occurs.
We organize the correction in terms of factorial cumulants: the first
describes the mean peak count, the second the connected pair
correlation, and higher cumulants encode genuinely higher-order
clustering.  
For the one-dimensional Gaussian matched-filter field we calculate the
second factorial cumulant using the two-point Kac--Rice formula for the
joint occurrence of pairs of local maxima~\cite{AzaisWschebor2009,AdlerTaylor2007},
and find a quartic short-distance exclusion law for nearby maxima.

The paper is organized as follows.  Section~2 reviews the literature
most directly related to nuisance parameters present only under the
alternative, random-field power calculations, optimization bias, and
Kac--Rice peak inference.  Section~3 isolates the common dimensional
structure connecting local profiling and the look-elsewhere effect.
Sections~4 and~5 introduce the Gaussian template geometry and derive
the strong-signal local results, while Sec.~6 establishes the
asymptotic differential bridge to the null trials factor.
Section~7 develops one- and two-dimensional Gaussian examples and
their model-dependent higher-order corrections, which are tested
numerically in Sec.~8.  Section~9 follows the transition from local
profiling to global competition between the signal-associated maximum
and remote extrema.  Section~10 develops finite-threshold corrections
to the null extrema process in terms of factorial cumulants, and
Sec.~11 derives the second cumulant from two-point Kac--Rice theory and
analyzes the short-distance exclusion of nearby maxima.  Sec.~12 summarizes the main conclusions and discusses limitations and
possible extensions.  Appendix~A derives the
higher-order local expansions used in Sec.~7, while Appendix~B gives
the principal-coordinate calculation underlying the short-distance
two-point result.

\section{Related work}
\label{sec:related}

The statistical ingredients entering the present problem have appeared
in several literatures, often with different terminology and with the
null and alternative hypotheses treated separately.  We summarize here
the strands that are most directly relevant to the connection pursued
in this paper.

\subsection{Nuisance parameters present only under the alternative \label{sec:related-nuisance}}

The nonregular testing problem produced by a signal location that is
undefined under the null was formulated explicitly by Davies
\cite{Davies1977,Davies1987}.  In that setting, a simple null hypothesis
is tested against a family of alternatives indexed by a nuisance
parameter that disappears when the signal amplitude vanishes.  Davies
reduced the problem to the extrema of Gaussian or chi-square processes
and derived approximations both for significance and, already in the
1977 treatment, for power.

In high-energy physics this problem became widely known as the
look-elsewhere effect. We quantify this multiplicity by the null trials factor
${\rm TF}_0(u)$ at local threshold $u$, defined conceptually as the
ratio between the background-only probability of exceeding $u$
somewhere in the search and the corresponding fixed-template tail
probability.  Its precise definition for the Gaussian field used below
is given in Sec.~\ref{sec:bridge}.
Gross and Vitells \cite{GrossVitells2010},
building directly on Davies, gave a practical high-threshold
approximation for resonance searches and showed that in one search
dimension the trials factor grows asymptotically linearly with the local
significance.  Their interpretation is especially relevant here: the
search range may be regarded as containing an effective number of
independent regions, while continuous optimization of the signal
location contributes an additional local degree of freedom inside each
region.  In their notation this appears through a high-tail combination
of an $s$-degree-of-freedom fixed-location term and an
$(s+1)$-degree-of-freedom search term.  Thus the distinction between
\emph{global multiplicity} and \emph{local continuous optimization} is
already present explicitly in the one-dimensional LEE treatment.
Vitells and Gross subsequently extended the analysis to multidimensional
search spaces using random-field and Euler-characteristic methods
\cite{VitellsGross2011}.  More generally, the geometry of excursion sets
and maxima of smooth random fields provides the mathematical framework
in which these high-threshold results are naturally expressed.

A complementary interpretation was developed by Bayer and Seljak
\cite{BayerSeljak2020}.  They related the look-elsewhere trials factor to
a prior-to-posterior volume ratio, providing a continuous version of the
usual multiplicity correction.  Bayer, Seljak and Robnik later proposed
a self-calibrating construction based on the heights of noise-induced
likelihood peaks \cite{BayerSeljakRobnik2021}.  The volume-ratio picture
is especially relevant here because, for a $D$-dimensional locally
regular signal manifold, the localization covariance shrinks as
$Q^{-2}$ and hence the local posterior volume scales as $Q^{-D}$.
This gives another route to the familiar high-significance scaling
${\rm TF}_0(u)\propto u^D$.

\subsection{Gaussian random fields under a genuine signal}

A particularly close antecedent to the present problem is the work of
Siegmund and Worsley \cite{SiegmundWorsley1995}.  They considered a
known signal shape embedded in stationary Gaussian noise, with unknown
amplitude, location, and scale, and maximized the corresponding Gaussian
random field over the unknown template coordinates.  Their paper is notable because it treats both sides of the problem:
under the null, the probability of a high maximum is approximated using
tube methods---which measure a narrow geometric neighborhood around the
template manifold---and related Euler-characteristic methods, while
the power is studied under a genuine signal.

In the strong-signal regime, Siegmund and Worsley's power calculation expands the field quadratically around the true signal point and maximizes over the local
location--scale coordinates.  The resulting optimization increment is
controlled by a quadratic form in the field gradient which, after
whitening, has a chi-square distribution with a number of degrees of
freedom equal to the number of profiled template coordinates.  In the
notation used below, this is the same local mechanism behind
\begin{equation}
  \widehat Q-Z_{\rm fix}
  \simeq
  \frac{\chi_D^2}{2Q},
  \qquad
  \Delta q\simeq\chi_D^2,
  \label{eq:rw-siegmund-local}
\end{equation}
at leading order.  Thus the local positive shift of the optimized peak
height is a known statistical phenomenon; an equivalent quadratic
maximization mechanism is already implicit in the alternative-side
analysis of Ref.~\cite{SiegmundWorsley1995}.

The importance of the above antecedent for this work is as follows. Both sides of the present treatment are already contained in
Ref.~\cite{SiegmundWorsley1995}: tube/Euler-characteristic calculations
control the high null maximum, while a local quadratic maximization
controls power under a genuine signal.  The two calculations are used for different inferential purposes---significance and power, respectively. To our knowledge, that work does not turn them into an explicit null--alternative identity involving the global trials factor, nor does it formulate the logarithmic-derivative relation derived below.  Our bridge should therefore be viewed as a synthesis of asymptotic structures that are already substantially present in this earlier work.

\subsection{Optimization bias and selected peak heights}

Positive bias caused by optimizing a noisy matched-filter statistic has
also been recognized directly in astronomical applications.  In the
South Pole Telescope cluster analysis of Vanderlinde et al.
\cite{Vanderlinde2010}, the detection significance was optimized over
two position coordinates and one angular-scale parameter.  A correction
equivalent at high signal-to-noise to a $3/(2Q)$ upward shift was
introduced heuristically.  Zubeldia et al. later studied this effect
systematically under the name \emph{optimization bias}
\cite{Zubeldia2021}.  
They showed that the optimized matched-filter signal-to-noise is
positively biased even for a perfectly specified signal in Gaussian
noise, and obtained, in their notation, an approximate relation of the
form
\begin{equation}
  \bar q_{\rm opt}
  \simeq
  \sqrt{\bar q_t^{\,2}+f_{\rm eff}},
  \label{eq:zubeldia}
\end{equation}
where their $\bar q$ denotes an SNR-like quantity and should not be
confused with the twice-log-likelihood-ratio statistic $q$ used in the
present paper.  Identifying the unoptimized signal-to-noise
$\bar q_t$ with our $Q$ and taking $f_{\rm eff}=D$, the large-$Q$
expansion gives
\[
  \bar q_{\rm opt}
  \simeq
  \sqrt{Q^2+D}
  =
  Q+\frac{D}{2Q}+O(Q^{-3}),
\]
so that the optimization bias
$\bar q_{\rm opt}-\bar q_t$ is $D/(2Q)$ at leading order.

Other applications display a more explicitly global form of
maximization bias.  In a search for millimeter transients, Whitehorn et
al. \cite{Whitehorn2016} noted that, without a compensating penalty, the
fit is biased toward short flare durations because shorter templates
permit more effectively independent start times and hence a larger
trials factor.  This is conceptually distinct from the strong-signal
local bias studied here: it is a bias in a fitted shape parameter driven
by variation of the \emph{global} search multiplicity.  The distinction
will be useful below when separating local profiling from remote-peak
competition.

A related literature studies inference after selecting local maxima.
Davenport and Nichols~\cite{DavenportNichols2020}, for example, showed
that effect sizes reported at neuroimaging peaks are biased upward
because the reported locations have both passed a threshold and been
selected as local maxima. This selective-inference perspective establishes that upward peak-height bias is a general phenomenon; our interest is in the analytic profiling
geometry of a signal template and its connection to the look-elsewhere
asymptotics.

\subsection{Non-centered random fields, power, and peak inference}

Random-field methods have also been developed explicitly under the
alternative.  Hayasaka et al. \cite{Hayasaka2007} introduced a
non-central random-field framework for power and sample-size calculations
in neuroimaging, treating anticipated signal regions while retaining the
spatial multiplicity correction.

More recently, Zhao, Cheng and Schwartzman \cite{ZhaoChengSchwartzman2024}
studied the power of detecting peaks in a non-centered Gaussian random
field.  They approximate the probability of at least one local maximum
above a threshold by the expected number of such maxima, calculated by
Kac--Rice methods, and establish the approximation in several regimes
including high thresholds and sharp signals.  Cheng
\cite{Cheng2025} derived exact formulas for the expected number and
height distribution of local maxima in classes of smooth non-centered
Gaussian fields, including stationary planar fields with deterministic
trends.  These developments provide tools directly relevant to the
weakening-signal regime in which a signal-associated maximum begins to
compete with unrelated extrema.

Green and Taylor \cite{GreenTaylor2025} provide the closest modern
contact with the matrix formulation used below.  They study inference for the location and height of selected peaks of a smooth
signal-plus-Gaussian-noise field and derive Kac--Rice approximations that are second-order accurate near high-curvature true peaks.  Their
alternative-side peak-height correction contains the local combination
\begin{equation}
  \frac12\operatorname{tr}(H^{-1}\Lambda),
  \label{eq:green-taylor-trace}
\end{equation}
where $H$ is the deterministic peak-curvature matrix and $\Lambda$ is the noise-gradient covariance.  They also discuss the corresponding null peak geometry.  Conditional on a stationary noise peak having a large height $v$, its mean curvature grows proportionally to $v\Lambda$, where $\Lambda$ is the covariance matrix of the field gradient.  The Kac--Rice density of local maxima contains the determinant of this curvature matrix; since
\[
  \det(v\Lambda)=v^D\det\Lambda,
\]
a $D$-dimensional search acquires the same $v^D$ dependence that
ultimately produces the high-threshold random-field trials factor.
Thus Green and Taylor place the local null and alternative peak geometries very nearly side by side.  In the normalized matched-filter problem considered here, $H_Q=QG$ and $\Lambda=G$, so
Eq.~\eqref{eq:green-taylor-trace} reduces to $D/(2Q)$.
What is not formulated there, to our knowledge, is a relation to the
\emph{global} look-elsewhere trials factor in the Gross--Vitells sense.
The step taken below is to combine the local Kac--Rice curvature identity
with the high-threshold global excursion asymptotics and to express the
result as a logarithmic derivative of the trials factor.

\subsection{Higher-order Kac--Rice structure}

The one-point Kac--Rice formula gives the mean number of extrema in a
search region.  Its higher-order versions give joint intensities and
factorial moments of critical-point counts; standard treatments may be
found, for example, in Refs.~\cite{AzaisWschebor2009,AdlerTaylor2007}.

The distinction between the mean density of extrema and their
higher-order correlation structure becomes important away from the
asymptotic high-threshold regime.  If $N_u$ is the number of local
maxima above a threshold, a Poisson approximation is determined
entirely by $\E[N_u]$, whereas the exact probability of finding no such
maximum also depends on correlations among the extrema, encoded by
higher factorial cumulants.
The present work uses the two-point Kac--Rice density to calculate the
second factorial cumulant explicitly for a one-dimensional Gaussian
matched-filter field.  This provides a controlled finite-threshold
correction to the leading Poisson picture rather than another effective
trials-factor parametrization.

\subsection{Position of the present work \label{s:position}}

Several individual ingredients of the argument below are therefore known.  In particular, we do not claim novelty for positive optimization
bias itself, the local $\chi^2_D$ profile gain, the coefficient
$D/(2Q)$ in isolation, or the high-threshold scaling
${\rm TF}_0(u)\propto u^D$.
Siegmund--Worsley,
Gross--Vitells, and Green--Taylor in particular contain remarkably close pieces of the structure developed below.

Our first contribution is to make the \emph{global null / local
alternative} relation explicit.  In the normalized matched-filter
problem, the same template metric controls the deterministic curvature
of a strong signal peak, the covariance of the random field gradient,
and the conditional curvature of a high null excursion.  Consequently
the same $D$ local profiling directions control
\begin{equation}
  {\rm TF}_0(u)\propto u^D,
  \qquad
  \E[\Delta q]\to D,
  \qquad
  2Q\,\E[\widehat Q-Q]\to D.
  \label{eq:related-three-D}
\end{equation}
Under the assumptions made explicit in Sec.~\ref{sec:bridge}, this
correspondence gives
\begin{equation}
  \E_Q[\widehat Q-Q]
  =
  \left.
  \frac12\frac{d}{du}\log {\rm TF}_0(u)
  \right|_{u=Q}
  +o(Q^{-1}).
  \label{eq:related-logderivative}
\end{equation}
The logarithmic derivative removes the threshold-independent global
search-volume factor and retains the local continuous-optimization
contribution.  We have not found this global-trials-factor differential
identity stated explicitly in the literature reviewed above.

Our second contribution is to continue the comparison beyond the strict
high-excursion Poisson limit.  We separate the signal-associated maximum
from competing extrema, quantify the small dependence corrections, and
show that finite-threshold deviations of the null void probability are
organized naturally by factorial cumulants.  For the clean
one-dimensional interior-maxima process, the leading non-Poisson
correction is then obtained from a two-point Kac--Rice calculation and
validated against simulation.  This part of the analysis distinguishes
three levels that should not be mixed up: local dimension, local
curvature, and the correlated global process of extrema.

\section{Local optimization and global search}
\label{sec:concept}

Before introducing the detailed model, it is useful to isolate the
structural observation that motivates the paper.

\subsection{Two components of the look-elsewhere effect}

In the one-dimensional resonance-search problem, Gross and Vitells show
that the high-threshold trials factor grows asymptotically linearly with
the local significance \cite{GrossVitells2010}.  
In a $D$-dimensional smooth search, the leading interior random-field
term has the corresponding power-law form
\begin{equation}
  {\rm TF}_0(u)
  \sim
  K_{\mathcal M}\,u^D,
  \qquad
  u\gg1,
  \label{eq:TF-scaling}
\end{equation}
where $K_{\mathcal M}$ contains the global information about the size and geometry
of the search region.

Following the interpretation emphasized by Gross and Vitells, this
scaling separates naturally into two ingredients:
the constant $K_{\mathcal M}$, which encodes the global multiplicity or search
volume, and the factor $u^D$, which reflects the continuous optimization
over $D$ local signal coordinates.  The former answers the familiar
question of how many effectively distinct regions were searched; the
latter quantifies the additional gain obtained by adjusting the
position or shape of a candidate excess within each such region.

\subsection{The same dimension on the alternative side}

The calculations developed below show that, under a strong signal, the
same number $D$ controls the local likelihood gain:
\begin{equation}
  \Delta q
  \equiv
  q_{\rm prof}-q_{\rm fixed}
  \xrightarrow{d}
  \chi^2_D,
  \qquad
  \E[\Delta q]\to D.
  \label{eq:concept-deltaq}
\end{equation}
The corresponding fitted-amplitude bias is
\begin{equation}
  \E[\widehat Q-Q]
  =
  \frac{D}{2Q}
  +o(Q^{-1}).
  \label{eq:concept-bias}
\end{equation}
Thus the same integer $D$ appears in three different objects:
\begin{equation}
  \begin{array}{ccl}
{\rm TF}_0(u) &\propto& u^D,
  \\[1mm]
  \E[\Delta q] &\longrightarrow& D,
  \\[1mm]
  2Q\,\E[\widehat Q-Q] &\longrightarrow& D.
  \end{array}
  \label{eq:three-Ds}
\end{equation}

It is useful here to distinguish two contributions to the fitted
amplitude.  At the true fixed template,
\[
  Z_{\rm fix} \equiv Z(\theta_0)=Q+X,
  \qquad
  \E[X]=0 ,
\]
so that
\[
  \widehat Q-Q
  =
  \underbrace{(Z_{\rm fix}-Q)}_{\text{radial noise}}
  +
  \underbrace{(\widehat Q-Z_{\rm fix})}_{\text{profiling gain}}.
\]
The first term is the ordinary fixed-template fluctuation and has zero
mean.  The positive greedy bias is generated by the second term, which
is the increase obtained by moving along the template manifold.
Consequently
\[
  \E[\widehat Q-Q]
  =
  \E[\widehat Q-Z_{\rm fix}].
\]

The last two lines of Eq.~\eqref{eq:three-Ds} are connected directly by
\begin{equation}
  \Delta q
  =
  2Q\,(\widehat Q-Z_{\rm fix})+o_{\mathrm p}(1),
  \label{eq:deltaq-bias-relation}
\end{equation}
where $o_{\mathrm p}(1)$ denotes a random remainder that converges to zero
in probability as $Q\to\infty$.  Thus the local likelihood gain is
converted into a positive amplitude shift by the slope
$dq/dQ\simeq2Q$ of the strong-signal likelihood ratio.

\subsection{A logarithmic-derivative correspondence}

Combining Eq.~\eqref{eq:TF-scaling} with
Eq.~\eqref{eq:concept-bias} gives
\begin{equation}
  \left.
  \frac12\frac{d}{du}\log {\rm TF}_0(u)
  \right|_{u=Q}
  =
  \frac{D}{2Q}
  \sim
  \E[\widehat Q-Q].
  \label{eq:logTF-bias}
\end{equation}
At this stage Eq.~\eqref{eq:logTF-bias} could appear to be only a
dimensional coincidence.  
The reason it is potentially more than that is that the exponent $D$
in the null trials factor and the coefficient $D$ in the strong-signal
bias originate from the same local tangent-space geometry.
Under a genuine signal, the deterministic matched-filter
peak becomes narrower in the $D$ tangent directions in proportion to
$Q$.  Under the null, conditioning on a rare peak of height $u$
produces the same height-dependent curvature in those directions.  The
alternative calculation therefore involves the inverse curvature
through a trace, whereas the null Kac--Rice calculation involves the
curvature through a determinant.  Section~\ref{sec:bridge} makes this
common-matrix statement precise and derives the logarithmic-derivative
relation under standard high-excursion regularity conditions.

The logarithmic derivative removes the constant $\log K_{\mathcal M}$, and
therefore eliminates the global search-volume factor while retaining
the local optimization term $D\log u$; evaluating the result at
$u=Q$ then connects it to the strong-signal bias.
This mirrors the strong-signal problem, in which the real
signal anchors the fit near a particular point of the template manifold:
the leading greedy bias depends on the local dimension but not on the
total volume searched.

The geometric derivation in the next sections explains why this same
$D$ arises on the alternative side.  Section~\ref{sec:bridge} goes one
step further and shows, within the matched-filter Gaussian model, that
the logarithmic-derivative relation follows from the same local
covariance/curvature matrix on the two sides of the hypothesis test.

\section{Gaussian template model and signal-cone geometry}
\label{sec:model}

\subsection{Whitened observation model}

For the derivation below we assume Gaussian noise explicitly, in order
to isolate the profiling geometry as cleanly as possible.  The
construction should be viewed as the whitened, quadratic limit of a
more general likelihood problem: after expansion around the true model,
the covariance is transformed to the identity and the signal templates
become vectors in the resulting data space.
In a counting experiment the underlying likelihood may,
for example, be Poisson or extended, but sufficiently regular
large-sample likelihoods have the same local quadratic structure.
The leading tangent-space results derived below are therefore expected
to be more general than the exact finite-threshold random-field formulas,
which depend on the Gaussian model assumed here.

Let the whitened data vector be
\begin{equation}
  \bm x
  =
  Q\,\bm s(\thetao)+\bm\epsilon,
  \qquad
  \bm\epsilon\sim N(\bm 0,I),
  \label{eq:data-model}
\end{equation}
where $\bm\epsilon$ is the whitened noise vector, with zero mean and
unit covariance, and $\bm s(\bm\theta)$ is a smooth family of signal
templates parametrized by
$\bm\theta=(\theta_1,\ldots,\theta_D)$ in the $D$-dimensional template
manifold $\mathcal M$, and normalized according to
\begin{equation}
  \ip{\bm s(\bm\theta)}{\bm s(\bm\theta)}=1,
  \qquad
  \bm\theta\in\mathcal M\subset\R^D.
  \label{eq:normalization}
\end{equation}
Here normalization refers to unit norm in the whitened data space
($\sum_a s_a^2=1$ for discretized data), rather than to unit integral
of the signal shape.  With this convention the matched-filter output
at any fixed template has unit noise variance.
The parameter $Q>0$ is the true signal amplitude measured in units of
the fixed-template standard deviation.

For fixed $\bm\theta$, the maximum-likelihood amplitude in this model is
simply the matched-filter field
\begin{equation}
  Z(\bm\theta)
  =
  \ip{\bm x}{\bm s(\bm\theta)}.
  \label{eq:Zfield}
\end{equation}
To connect this notation with the usual likelihood-ratio language,
consider the Gaussian log likelihood
\begin{equation}
  \ell(\mu,\bm\theta)
  =
  -\frac12
  \left\|
    \bm x-\mu\,\bm s(\bm\theta)
  \right\|^2
  +{\rm const}.
  \label{eq:gaussian-loglik}
\end{equation}
For fixed $\bm\theta$, maximizing over an unrestricted amplitude gives
\begin{equation}
  \widehat\mu(\bm\theta)=Z(\bm\theta).
\end{equation}
The twice-log-likelihood ratio relative to $\mu=0$ is therefore
\begin{equation}
  q(\bm\theta)
  =
  2\bigl[
    \ell(\widehat\mu,\bm\theta)-\ell(0)
  \bigr]
  =
  Z(\bm\theta)^2.
  \label{eq:q-Z}
\end{equation}
If the physical signal strength is constrained to $\mu\ge0$, the
corresponding statistic is $q(\bm\theta)=Z_+(\bm\theta)^2$, with
$Z_+=\max(0,Z)$.  In the strong-signal regime considered in the next
sections the probability of encountering this boundary is exponentially
small, so the two forms are asymptotically equivalent.
We denote by $\widehat Q$ the value of this field at the local maximum
associated with the true signal,
\begin{equation}
  \thetahat
  =
  \arg\max_{\bm\theta\approx\thetao}Z(\bm\theta),
  \qquad
  \widehat Q=Z(\thetahat).
  \label{eq:localmax}
\end{equation}
The restriction to the signal-associated local maximum is essential in
the present section: global competition with remote maxima belongs to a
different, weak-signal regime.

\subsection{The signal cone}

The strict positive-signal alternative is
\begin{equation}
  \mathcal C^\circ
  =
  \{\mu\,\bm s(\bm\theta):
    \mu>0,\ \bm\theta\in\mathcal M\}.
\end{equation}
Its closure,
\begin{equation}
  \mathcal C
  =
  \{\mu\,\bm s(\bm\theta):
    \mu\ge0,\ \bm\theta\in\mathcal M\},
  \label{eq:cone}
\end{equation}
is a cone over the template manifold.  The point $\mu=0$ is the cone
apex and represents the null hypothesis.

This picture makes the regularity issue transparent.  For $\mu>0$,
changes in $\bm\theta$ move the mean prediction:
\begin{equation}
  \frac{\partial}{\partial\theta_i}
  \bigl[\mu\,\bm s(\bm\theta)\bigr]
  =
  \mu\,\partial_i\bm s(\bm\theta).
\end{equation}
At $\mu=0$, all these derivatives vanish and the template coordinates
are unidentifiable.  
Thus the strong-signal problem is locally regular away from the apex,
while the null problem is singular.  In the former regime the usual Wilks picture for nested regular models applies
locally~\cite{Wilks1938,CowanEtAl2011}: freeing $D$
additional identifiable template coordinates produces an asymptotic
$\chi^2_D$ likelihood gain.  At the apex, however,
those coordinates are not identifiable under the null, violating the
regularity assumptions of Wilks' theorem and leading instead to the
Davies/random-field problem underlying the look-elsewhere
effect~\cite{Davies1977,Davies1987}.

\subsection{Local metric}

At the true template define
\begin{equation}
  \bm s_0=\bm s(\thetao),
  \qquad
  \bm s_i=\partial_i\bm s(\thetao),
\end{equation}
and the $D\times D$ metric
\begin{equation}
  G_{ij}
  =
  \ip{\bm s_i}{\bm s_j}.
  \label{eq:metric}
\end{equation}
Normalization implies
\begin{equation}
  \ip{\bm s_0}{\bm s_i}=0,
  \label{eq:orthogonal}
\end{equation}
so the radial amplitude direction is orthogonal to all tangent
directions.  Differentiating the normalization once more gives
\begin{equation}
  \ip{\bm s_0}{\partial_i\partial_j\bm s}
  =
  -G_{ij}.
\end{equation}
Consequently, for
$\dtheta=\bm\theta-\thetao$,
\begin{equation}
  \rho(\dtheta)
  \equiv
  \ip{\bm s_0}{\bm s(\thetao+\dtheta)}
  =
  1-\frac12\dtheta^T G\dtheta
  +O(\|\dtheta\|^3).
  \label{eq:overlap-expansion}
\end{equation}

Thus $G$ measures how rapidly pairs of neighboring normalized signal templates become distinguishable.  In local coordinates,
$\delta\bm\theta^T G\,\delta\bm\theta$ is the squared mismatch between
the true template and a nearby one to leading order.  Equivalently, the expression plays here the role of the Fisher-information metric for the profiled
template coordinates, after the radial amplitude direction has been
separated and the data have been whitened.
\section{Local profiling of a strong signal}
\label{sec:local}

We now state the leading local result.  The assumptions are:
(i) Gaussian whitened noise as in Eq.~\eqref{eq:data-model};
(ii) a normalized template family that is three times continuously
differentiable ($C^3$) in the template coordinates;
(iii) a true parameter point \(\theta_0\) lying in the interior of the search region; and
(iv) a positive-definite local metric $G$.
The calculation separates two kinds of noise fluctuation.  The scalar component along the true template changes the fitted amplitude even when the template is held fixed.  The $D$ components along the tangent directions instead move the fitted template away from the true point. The latter displacements are the source of the additional positive profiling gain.  The proposition below makes this decomposition
explicit.

\begin{proposition}[Local profiling in $D$ dimensions]
\label{prop:local}
Let $Q\to\infty$ with the signal-associated local maximum selected as in
Eq.~\eqref{eq:localmax}.  Define
\begin{equation}
  X=\ip{\bm\epsilon}{\bm s_0},
  \qquad
  g_i=\ip{\bm\epsilon}{\bm s_i}.
\end{equation}
Then
\begin{equation}
  X\sim N(0,1),
  \qquad
  \bm g\sim N(\bm0,G),
  \qquad
  X\perp\bm g,
  \label{eq:Xg}
\end{equation}
and
\begin{align}
  \thetahat-\thetao
  &=
  \frac{1}{Q}G^{-1}\bm g
  +\op(Q^{-2}),
  \label{eq:theta-hat-leading}
  \\
  \widehat Q
  &=
  Q+X
  +\frac{1}{2Q}\bm g^T G^{-1}\bm g
  +\op(Q^{-2}).
  \label{eq:Qhat-leading}
\end{align}
Since
\begin{equation}
  \bm g^T G^{-1}\bm g\sim\chi^2_D,
\end{equation}
the leading amplitude bias is
\begin{equation}
  \E[\widehat Q-Q]
  =
  \frac{D}{2Q}
  +o(Q^{-1}).
  \label{eq:bias-leading}
\end{equation}
\end{proposition}

\begin{proof}
Write
\begin{equation}
  Z(\thetao+\dtheta)
  =
  Q\,\rho(\dtheta)+W(\dtheta),
\end{equation}
where
\begin{equation}
  W(\dtheta)
  =
  \ip{\bm\epsilon}{\bm s(\thetao+\dtheta)}.
\end{equation}
Using Eq.~\eqref{eq:overlap-expansion} and expanding the noise field,
\begin{equation}
  W(\dtheta)
  =
  X+\bm g^T\dtheta+O_{\mathrm p}(\|\dtheta\|^2),
\end{equation}
we obtain
\begin{equation}
  Z(\thetao+\dtheta)
  =
  Q+X+\bm g^T\dtheta
  -\frac{Q}{2}\dtheta^T G\dtheta
  +O_{\mathrm p}(\|\dtheta\|^2)
  +Q\,O(\|\dtheta\|^3).
  \label{eq:Zlocal}
\end{equation}
The derivative of the linear noise term is $O_{\rm p}(1)$, whereas the restoring derivative of the deterministic quadratic term is $O_{\rm p}(Q\|\dtheta\|)$.  Balancing the two therefore gives
$\dtheta=O_{\rm p}(Q^{-1})$. Differentiating Eq.~\eqref{eq:Zlocal} then gives
\begin{equation}
  \bm0
  =
  \bm g-QG\dtheta+\op(Q^{-1}),
\end{equation}
which yields Eq.~\eqref{eq:theta-hat-leading}.  Substituting this solution
back into Eq.~\eqref{eq:Zlocal} gives
\begin{equation}
  \widehat Q
  =
  Q+X
  +\frac{1}{2Q}\bm g^TG^{-1}\bm g
  +\op(Q^{-2}).
\end{equation}

The covariance relations follow directly from the whitened Gaussian
noise:
\begin{align}
  \Var(X)
  &=
  \ip{\bm s_0}{\bm s_0}=1,
  \\
  \Cov(X,g_i)
  &=
  \ip{\bm s_0}{\bm s_i}=0,
  \\
  \Cov(g_i,g_j)
  &=
  \ip{\bm s_i}{\bm s_j}=G_{ij}.
\end{align}
Joint Gaussianity therefore implies independence of $X$ and $\bm g$.
Whitening the tangent gradient,
$\bm y=G^{-1/2}\bm g\sim N(\bm0,I_D)$, gives
\begin{equation}
  \bm g^TG^{-1}\bm g
  =
  \bm y^T\bm y
  \sim\chi^2_D,
\end{equation}
and Eq.~\eqref{eq:bias-leading} follows.
\end{proof}

\subsection{Profile-likelihood gain}

The likelihood-ratio interpretation makes the role of the tangent
directions particularly transparent.  A fixed-template fit uses only
the radial Gaussian fluctuation $X$.  Allowing the template coordinates
to move adds the $D$ tangent fluctuations through the positive quadratic
form $\bm g^TG^{-1}\bm g$.

At the true fixed template,
\begin{equation}
  q_{\rm fixed}
  =
  Z(\thetao)^2
  =
  (Q+X)^2,
  \label{eq:qfixed}
\end{equation}
up to an exponentially small correction from a one-sided
$\mu\ge0$ boundary at large $Q$.

Using Eq.~\eqref{eq:Qhat-leading},
\begin{equation}
  q_{\rm prof}
  =
  \widehat Q^2
  =
  (Q+X)^2
  +\bm g^TG^{-1}\bm g
  +\op(Q^{-1}).
\end{equation}
Therefore
\begin{equation}
  \Delta q
  \equiv
  q_{\rm prof}-q_{\rm fixed}
  =
  \bm g^TG^{-1}\bm g+\op(Q^{-1}).
  \label{eq:deltaq}
\end{equation}

\begin{corollary}[Extra local degrees of freedom]
\label{cor:deltaq}
In the strong-signal limit,
\begin{equation}
  \Delta q
  \xrightarrow{d}\chi^2_D,
  \qquad
  \E[\Delta q]\to D.
  \label{eq:deltaq-limit}
\end{equation}
Equivalently, the tangent-space approximation to the profiled statistic is
\begin{equation}
  q_{\rm prof}^{(\mathrm{tan})}
  =
  (Q+X)^2+\sum_{a=1}^D Y_a^2,
  \qquad
  X,Y_a\stackrel{\rm iid}{\sim}N(0,1),
\end{equation}
so that
\begin{equation}
  q_{\rm prof}^{(\mathrm{tan})}
  \sim
  \chi'^2_{D+1}(Q^2).
  \label{eq:noncentral}
\end{equation}
\end{corollary}

The result may be read geometrically.  The radial fluctuation $X$ is
present already in a fixed-template fit.  Profiling over $D$ local
template coordinates makes $D$ additional orthogonal Gaussian
directions available.  Their contribution to the likelihood ratio enters as a sum of squares and is therefore positive.  This is the local origin of the greedy bump bias.

\subsection{Localization of the fitted template}

Equation~\eqref{eq:theta-hat-leading} shows that the displacement of
the fitted template from the true point is of order $Q^{-1}$, with
asymptotic covariance
\[
  \Cov(\thetahat-\thetao)
  \simeq
  \frac{1}{Q^2}G^{-1}.
\]
A natural dimensionless measure of the localization error is therefore
the squared distance in the local template metric,
\begin{equation}
  r_{\rm loc}^2
  \equiv
  Q^2(\thetahat-\thetao)^T
  G
  (\thetahat-\thetao).
  \label{eq:rloc2-def}
\end{equation}

To express this distance in Cartesian coordinates, let $G=LL^T$ be a
Cholesky decomposition and define
\begin{equation}
  \bm y_{\rm loc}
  =
  QL^T(\thetahat-\thetao).
  \label{eq:whitened-localization}
\end{equation}
Then $r_{\rm loc}^2=\bm y_{\rm loc}^T\bm y_{\rm loc}$, and
Eq.~\eqref{eq:theta-hat-leading} implies
\begin{equation}
  \bm y_{\rm loc}
  \xrightarrow{d}
  N(\bm0,I_D),
\end{equation}
so that
\begin{equation}
  r_{\rm loc}^2\xrightarrow{d}\chi^2_D.
  \label{eq:rloc2}
\end{equation}

\section{A local bridge to the high-threshold trials factor}
\label{sec:bridge}

The preceding calculation recovers the alternative-side local geometry.
Its quadratic optimization mechanism is already present in
Siegmund--Worsley \cite{SiegmundWorsley1995}, and the general trace form
of the local peak-height correction is explicit in the recent Kac--Rice
analysis of Green--Taylor \cite{GreenTaylor2025}.  The high-threshold
null excursion scaling is likewise standard random-field theory and is
the basis of the Gross--Vitells treatment \cite{GrossVitells2010}.
Here we combine these ingredients in the normalized matched-filter
geometry and show that they imply an explicit differential relation
between the strong-signal optimization bias and the global null trials
factor.

\subsection {Bridging the two regimes}

For clarity, distinguish the null field
\begin{equation}
  Z_0(\bm\theta)
  =
  \ip{\bm\epsilon}{\bm s(\bm\theta)}
\end{equation}
from the field under a signal of strength $Q$ injected at $\thetao$,
\begin{equation}
  Z_Q(\bm\theta)
  =
  Q\,\rho(\bm\theta,\thetao)+Z_0(\bm\theta),
  \qquad
  \rho(\bm\theta,\thetao)
  =
  \ip{\bm s(\bm\theta)}{\bm s(\thetao)}.
  \label{eq:bridge-fields}
\end{equation}
The covariance of the null-field gradient is
\begin{equation}
  \Cov\!\left(
    \partial_i Z_0,\partial_j Z_0
  \right)
  =
  G_{ij}.
  \label{eq:bridge-gradient-cov}
\end{equation}

The central geometric observation can already be stated at this point.
Under a genuine signal, the deterministic mean field has negative
Hessian
\begin{equation}
  H_Q
  \equiv
  -\nabla^2 \E_Q[Z_Q(\bm\theta)]\big|_{\bm\theta=\thetao}
  =
  QG.
  \label{eq:bridge-preview-HQ}
\end{equation}
Under the null, conditioning on a high stationary excursion of height
$u$ produces, to leading order,
\begin{equation}
  \E\!\left[
    -\nabla^2 Z_0
    \mid
    Z_0=u,\nabla Z_0=0
  \right]
  =
  uG.
  \label{eq:bridge-preview-null}
\end{equation}
Thus the same metric $G$ controls the local curvature on the two sides
of the hypothesis test.

The two calculations use this curvature in different ways.  Under the
alternative, maximizing a quadratic peak produces an inverse-curvature
factor and hence a trace,
\begin{equation}
  \frac12\operatorname{tr}\!\left[(QG)^{-1}G\right]
  =
  \frac{D}{2Q}.
  \label{eq:bridge-preview-trace}
\end{equation}
Under the null, the Kac--Rice density of a high maximum contains the
curvature determinant,
\begin{equation}
  \det(uG)=u^D\det G.
  \label{eq:bridge-preview-det}
\end{equation}
The logarithmic derivative of the resulting $u^D$ factor is precisely
what reproduces the $D/(2Q)$ shift.  The proposition below turns this
matrix correspondence into the stated asymptotic relation.

Let
\begin{equation}
  M_0=\sup_{\theta\in\mathcal M} Z_0(\theta),
\end{equation}
where $M_0$ is the largest value of the matched-filter field found
anywhere in the search under the background-only hypothesis.  Denoting
probabilities under this null hypothesis by $\Pr_0$, the null trials
factor introduced in Sec.~\ref{sec:related-nuisance} is, more precisely,
\begin{equation}
  {\rm TF}_0(u)
  \equiv
  \frac{\Pr_0(M_0>u)}
       {\bar\Phi(u)},
  \qquad
  \bar\Phi(u)=1-\Phi(u),
  \label{eq:TF-definition}
\end{equation}
where $\bar\Phi(u)$ is the corresponding one-sided tail probability
at a single fixed template.

In this normalization the Gaussian threshold $u$ is also the usual
high-significance local significance, since the one-sided
likelihood-ratio statistic satisfies $q=Z_+^2$, with
$Z_+=\max(0,Z)$.  A threshold $u$ in the field therefore corresponds
to $q=u^2$.

We assume in this section that $\mathcal M$ is a smooth
$D$-dimensional search manifold, that $G(\bm\theta)$ is positive
definite, and that the standard smoothness and non-degeneracy conditions
required for a Kac--Rice high-excursion expansion hold.  We also assume
that the leading $D$-dimensional interior contribution is present.
Boundary terms, if present, are of lower dimension and do not change the
leading power of $u$.

The two component asymptotics in the proposition below should therefore
be regarded as established ingredients from the literature; the point of
the proposition is their combination into the global-null/local-alternative differential identity.

\begin{proposition}[Matched-filter bridge]
\label{prop:bridge}
Under the assumptions above, let $\widehat Q$ denote the height of the
signal-associated local maximum of $Z_Q$ for $Q\to\infty$.  Then
\begin{equation}
  \E_Q[\widehat Q-Q]
  =
  \frac{D}{2Q}
  +o(Q^{-1}).
  \label{eq:bridge-bias}
\end{equation}
Under the null, the high-threshold trials factor satisfies
\begin{equation}
  {\rm TF}_0(u)
  =
  K_{\mathcal M}\,u^D
  \left[1+o(1)\right],
  \qquad
  u\to\infty,
  \label{eq:bridge-TF}
\end{equation}
where $K_{\mathcal M}$ is independent of $u$.  If the high-threshold
expansion is differentiable at this order, then
\begin{equation}
  \boxed{
  \E_Q[\widehat Q-Q]
  =
  \left.
  \frac12\frac{d}{du}\log{\rm TF}_0(u)
  \right|_{u=Q}
  +o(Q^{-1}).
  }
  \label{eq:bridge-main}
\end{equation}
Moreover,
\begin{equation}
  \E_Q[\Delta q]
  =
  \left.
  Q\,\frac{d}{du}\log{\rm TF}_0(u)
  \right|_{u=Q}
  +o(1)
  =
  D+o(1).
  \label{eq:bridge-deltaq}
\end{equation}
\end{proposition}

\begin{proof}
\medskip\noindent\textit{Alternative side.}\quad
We first consider the alternative.  The deterministic mean field is
\begin{equation}
  m_Q(\bm\theta)
  =
  \E_Q[Z_Q(\bm\theta)]
  =
  Q\,\rho(\bm\theta,\thetao).
\end{equation}
Using Eq.~\eqref{eq:overlap-expansion}, its negative Hessian at the true
point is
\begin{equation}
  H_Q
  \equiv
  -\nabla^2 m_Q(\thetao)
  =
  QG.
  \label{eq:HQ-QG}
\end{equation}
Let
\begin{equation}
  \bm g
  =
  \nabla Z_0(\thetao),
  \qquad
  \bm g\sim N(\bm0,G).
\end{equation}
The local quadratic approximation reads
\begin{equation}
  Z_Q(\thetao+\dtheta)
  =
  Z_Q(\thetao)
  +\bm g^T\dtheta
  -\frac12\dtheta^T H_Q\dtheta
  +o_{\mathrm p}(Q^{-1}),
\end{equation}
uniformly on the $O_{\mathrm p}(Q^{-1})$ neighborhood relevant for the
local maximizer.  Maximization gives
\begin{equation}
  \widehat{\dtheta}
  =
  H_Q^{-1}\bm g
  +o_{\mathrm p}(Q^{-1})
\end{equation}
and therefore
\begin{equation}
  \widehat Q-Z_Q(\thetao)
  =
  \frac12
  \bm g^T H_Q^{-1}\bm g
  +o_{\mathrm p}(Q^{-1}).
  \label{eq:bridge-local-gain}
\end{equation}
Taking expectations and using
$\E[\bm g\bm g^T]=G$,
\begin{equation}
  \E_Q[\widehat Q-Z_Q(\thetao)]
  =
  \frac12
  \operatorname{tr}(H_Q^{-1}G)
  +o(Q^{-1})
  =
  \frac{D}{2Q}
  +o(Q^{-1}).
  \label{eq:bridge-trace}
\end{equation}
Since $\E_Q[Z_Q(\thetao)]=Q$, this proves
Eq.~\eqref{eq:bridge-bias}.

\medskip\noindent\textit{Null side.}\quad
We next turn to the null field.  Differentiating the unit-normalization
condition twice gives
\begin{equation}
  \Cov\!\left(
    \partial_i\partial_j Z_0,Z_0
  \right)
  =
  \ip{\partial_i\partial_j\bm s}{\bm s}
  =
  -G_{ij}.
  \label{eq:hessian-height-cov}
\end{equation}
Also
$\Cov(\nabla Z_0,Z_0)=0$.  Hence, conditional on a high field value
$Z_0=z$ and a stationary point $\nabla Z_0=0$,
\begin{equation}
  \E\!\left[
    -\nabla^2 Z_0
    \mid
    Z_0=z,\nabla Z_0=0
  \right]
  =
  zG.
  \label{eq:null-conditional-hessian}
\end{equation}
The residual conditional Hessian fluctuations remain $O_{\mathrm p}(1)$
as $z\to\infty$.  Therefore the determinant entering the Kac--Rice
density of high local maxima has leading behavior
\begin{equation}
  \det(-\nabla^2 Z_0)
  =
  z^D\det G\,[1+o_{\mathrm p}(1)].
  \label{eq:kac-rice-det}
\end{equation}

The gradient density at zero is
\begin{equation}
  p_{\nabla Z_0}(0)
  =
  (2\pi)^{-D/2}
  (\det G)^{-1/2}.
\end{equation}
Consequently, the leading interior Kac--Rice intensity of maxima above
$u$ is
\begin{align}
  \E_0[N_u]
  &\sim
  (2\pi)^{-D/2}
  \int_{\mathcal M}
    \sqrt{\det G(\bm\theta)}\,d\bm\theta
  \int_u^\infty
    z^D\phi(z)\,dz
  \nonumber\\
  &\equiv
  K_{\mathcal M}
  \int_u^\infty
    z^D\phi(z)\,dz.
  \label{eq:ENu}
\end{align}
The constant $K_{\mathcal M}$ contains the threshold-independent global
geometry of the search, in particular the search volume measured in
the local metric $G$, together with the associated normalization
factors.

For a smooth Gaussian field at high threshold, the probability of at
least one excursion maximum is asymptotically given by the leading
Kac--Rice/Euler-characteristic term, so
\begin{equation}
  \Pr_0(M_0>u)
  \sim
  K_{\mathcal M}\,
  u^{D-1}\phi(u).
  \label{eq:null-tail}
\end{equation}
Using the standard large-$u$ Gaussian-tail approximation
(Mills' ratio),
\begin{equation}
  \bar\Phi(u)
  \sim
  \frac{\phi(u)}{u},
  \qquad u\to\infty,
\end{equation}
we obtain
\begin{equation}
  {\rm TF}_0(u)
  \sim
  K_{\mathcal M}u^D,
\end{equation}
which proves Eq.~\eqref{eq:bridge-TF}.

Taking the logarithmic derivative gives
\begin{equation}
  \frac12\frac{d}{du}\log{\rm TF}_0(u)
  =
  \frac{D}{2u}
  +o(u^{-1}).
  \label{eq:bridge-log-derivative}
\end{equation}
Evaluating at $u=Q$ and comparing with
Eq.~\eqref{eq:bridge-bias} proves Eq.~\eqref{eq:bridge-main}.

Finally, from Eq.~\eqref{eq:bridge-local-gain},
\begin{equation}
  \Delta q
  =
  \widehat Q^2-Z_Q(\thetao)^2
  =
  2Q\,[\widehat Q-Z_Q(\thetao)]
  +o_{\mathrm p}(1).
\end{equation}
Taking expectations and using Eq.~\eqref{eq:bridge-main} yields
Eq.~\eqref{eq:bridge-deltaq}.
\end{proof}

\subsection{Why the relation is structural}

The proof makes the matrix correspondence previewed above explicit.
In the normalized matched-filter geometry,
\begin{equation}
  H_Q=QG,
  \qquad
  \Lambda\equiv\Cov(\nabla Z_0)=G,
\end{equation}
and hence
\begin{equation}
  \E_Q[\widehat Q-Q]
  =
  \frac12\operatorname{tr}(H_Q^{-1}\Lambda)
  +o(Q^{-1}).
\end{equation}
Because
\begin{equation}
  \frac{dH_Q}{dQ}=G=\Lambda,
\end{equation}
this can also be written
\begin{equation}
  \E_Q[\widehat Q-Q]
  =
  \frac12
  \frac{d}{dQ}\log\det H_Q
  +o(Q^{-1}).
  \label{eq:bias-logdet}
\end{equation}
The bridge described above therefore does not arise merely because two unrelated
formulas happen to contain the integer $D$: the same covariance metric
that controls random gradients under the null also controls the
curvature induced by the matched signal under the alternative.

\begin{remark}[Why this is special to matched likelihood geometry]
For a generic non-centered Gaussian field, let $H_Q$ denote the negative
Hessian of the deterministic mean at the true peak and let
$\Lambda=\Cov(\nabla Z_0)$ denote the noise-gradient covariance.  The
leading local peak-height shift has the general form
\begin{equation}
  B(Q)
  =
  \frac12\operatorname{tr}(H_Q^{-1}\Lambda).
\end{equation}
On the other hand,
\begin{equation}
  \frac12\frac{d}{dQ}\log\det H_Q
  =
  \frac12
  \operatorname{tr}
  \left(
    H_Q^{-1}\frac{dH_Q}{dQ}
  \right).
\end{equation}
The two coincide when
\begin{equation}
  \frac{dH_Q}{dQ}=\Lambda,
  \label{eq:matching-condition}
\end{equation}
or, more generally, when the two matrices have the same contraction with
$H_Q^{-1}$.  In the normalized matched-filter problem,
$H_Q=QG$ and $\Lambda=G$, so
Eq.~\eqref{eq:matching-condition} holds identically.
\end{remark}

\subsection{Localization-volume interpretation}

The leading covariance of the fitted template coordinates is
\begin{equation}
  \Cov(\thetahat)
  =
  \frac{1}{Q^2}G^{-1}
  +o(Q^{-2}).
\end{equation}
A local uncertainty volume therefore scales as
\begin{equation}
  V_{\rm loc}(Q)
  \propto
  [\det\Cov(\thetahat)]^{1/2}
  \propto
  Q^{-D}(\det G)^{-1/2}.
  \label{eq:local-volume}
\end{equation}
Hence
\begin{equation}
  -\frac12
  \frac{d}{dQ}\log V_{\rm loc}(Q)
  =
  \frac{D}{2Q}.
  \label{eq:bias-volume}
\end{equation}
The same shrinking localization cell is what makes the number of
effectively distinguishable search locations grow as $Q^D$ in the
high-threshold null problem.  In this sense,
Eq.~\eqref{eq:bridge-main} relates two views of the same local resolution
geometry: an upward optimization shift under the alternative and a
growing multiplicity of resolvable search locations under the null.

\section{Gaussian examples and higher-order corrections}
\label{sec:examples}

The leading results of Sec.~\ref{sec:local} depend only on the number
$D$ of locally identifiable template coordinates.  Model dependence
first enters through higher derivatives of the template manifold.
To separate these two levels explicitly, we consider two Gaussian
examples.  The first profiles only the signal location and will also
serve as the model for the local-to-global study below.  The second
profiles both location and width, providing a simple two-dimensional
manifold in which different tangent directions have the same leading
metric normalization but different higher-order corrections.

\subsection{One-dimensional location model}
\label{sec:d1}

Consider a normalized Gaussian template of fixed width $\sigma$ whose
location $m$ is profiled.  For two templates centered at $m_0$ and
$m=m_0+\delta m$, define their overlap as their inner product,
\begin{equation}
  \rho(\delta m)
  \equiv
  \ip{\bm s(m_0)}{\bm s(m)}
  =
  \exp\left[
    -\frac{(\delta m)^2}{4\sigma^2}
  \right].
  \label{eq:gaussian-overlap-d1}
\end{equation}
The local metric is therefore
\begin{equation}
  G
  =
  \frac{1}{2\sigma^2}.
\end{equation}
The leading predictions from Sec.~\ref{sec:local} become
\begin{align}
  \E[\widehat Q-Q]
  &=
  \frac{1}{2Q}+o(Q^{-1}),
  \label{eq:d1-bias-leading}
  \\
  \operatorname{RMS}(\widehat m-m_0)
  &=
  \frac{\sqrt{2}\,\sigma}{Q}
  +o(Q^{-1}),
  \label{eq:d1-mass-resolution}
  \\
  \Delta q
  &\xrightarrow{d}\chi^2_1.
  \label{eq:d1-deltaq}
\end{align}

To go beyond the universal tangent approximation, we expand the
stationarity equation for the profiled maximum and its height
iteratively in powers of $1/Q$.  Since the fitted displacement is
itself $O_{\rm p}(Q^{-1})$, higher derivatives of the template overlap
and of the Gaussian noise field generate successive curvature
corrections.  A derivation is given in
Appendix~\ref{app:local-higher-order}.  For the Gaussian location
manifold this gives
\begin{align}
  \E[\widehat Q-Q]
  &=
  \frac{1}{2Q}
  +\frac{3}{8Q^3}
  +\frac{25}{16Q^5}
  +O(Q^{-7}),
  \label{eq:d1-bias-nlo}
  \\
  \E[\Delta q]
  &=
  1+\frac{1}{2Q^2}
  +\frac{2}{Q^4}
  +O(Q^{-6}),
  \label{eq:d1-deltaq-nlo}
\end{align}
and, in the natural coordinate
\begin{equation}
  t=\frac{m-m_0}{\sqrt2\,\sigma},
\end{equation}
\begin{equation}
  \E[\widehat t^2]
  =
  \frac{1}{Q^2}
  +\frac{3}{Q^4}
  +\frac{45}{2Q^6}
  +O(Q^{-8}).
  \label{eq:d1-localization-nlo}
\end{equation}
This is precisely the metric-whitened coordinate
\[
  t=\sqrt{G}\,(m-m_0).
\]

The leading localization result becomes
\[
  Q\widehat t\xrightarrow{d}N(0,1).
\]

\subsection{Two-dimensional location--scale model}
\label{sec:d2}

We next allow both the position and the width of a Gaussian signal to
float.  It is convenient to parameterize the width by
\begin{equation}
  \eta=\log(\sigma/\sigma_0),
  \label{eq:eta}
\end{equation}
so that $\eta=0$ is an interior point and $\sigma>0$ automatically.
Taking $m_0=0$ and measuring $m$ in units of $\sigma_0$, the exact
normalized-template overlap with the true signal is
\begin{equation}
  \rho(m,\eta)
  =
  \frac{1}{\sqrt{\cosh\eta}}
  \exp\left[
    -\frac{m^2}{2(1+e^{2\eta})}
  \right].
  \label{eq:d2-overlap}
\end{equation}
At the true point,
\begin{equation}
  G
  =
  \begin{pmatrix}
    1/2 & 0\\
    0   & 1/2
  \end{pmatrix}.
  \label{eq:d2-metric}
\end{equation}
Thus the two profiled directions are indistinguishable at the level of
the leading tangent-space metric.  Any difference between their
finite-$Q$ behavior must therefore arise from higher derivatives of the
template manifold rather than from the leading local normalization.

The universal leading predictions are
\begin{align}
  \E[\widehat Q-Q]
  &=
  \frac{1}{Q}+o(Q^{-1}),
  \label{eq:d2-bias-leading}
  \\
  \Delta q
  &\xrightarrow{d}\chi^2_2,
  \label{eq:d2-deltaq-leading}
  \\
  \bm y_{\rm loc}
  &\xrightarrow{d}N(\bm0,I_2),
  \qquad
  r_{\rm loc}^2\xrightarrow{d}\chi^2_2.
  \label{eq:d2-localization-leading}
\end{align}

\subsubsection{Model-specific higher-order terms}

Applying the same higher-order expansion to the location--scale
manifold gives
\begin{align}
  \E[\widehat Q-Q]
  &=
  \frac{1}{Q}
  +\frac{9}{4Q^3}
  +O(Q^{-5}),
  \label{eq:d2-bias-nlo}
  \\
  \E[\widehat\eta]
  &=
  -\frac{1}{Q^2}
  +O(Q^{-4}),
  \label{eq:d2-eta-bias}
  \\
  \E[\Delta q]
  &=
  2+\frac{9}{2Q^2}
  +O(Q^{-4}),
  \label{eq:d2-deltaq-nlo}
  \\
  \E[r_{\rm loc}^2]
  &=
  2+\frac{35}{2Q^2}
  +O(Q^{-4}),
  \label{eq:d2-rloc2-nlo}
  \\
  \Var(y_m)
  &=
  1+\frac{6}{Q^2}
  +O(Q^{-4}),
  \label{eq:d2-var-m}
  \\
  \Var(y_\eta)
  &=
  1+\frac{11}{Q^2}
  +O(Q^{-4}).
  \label{eq:d2-var-eta}
\end{align}
The unequal coefficients in
Eqs.~\eqref{eq:d2-var-m} and~\eqref{eq:d2-var-eta} are therefore a
direct probe of higher-order manifold geometry rather than of the
leading metric.

\section{Numerical validation}
\label{sec:numerics}

\subsection{Simulation setup}

We tested the analytic results with ROOT pseudoexperiments using
independent unit Gaussian noise at each sampled data point and
discretely normalized Gaussian templates.  The fits were deliberately
local: they were initialized at the injected signal point and restricted
to a neighborhood intended to contain the signal-associated maximum,
rather than to perform a global look-elsewhere scan.

For the one-dimensional study we take $\sigma=1$ and sample the data at
101 equally spaced points over $[-6\sigma,6\sigma]$.  The local
location fit is restricted to $|m-m_0|\le2\sigma$ and evaluated on 81
scan points, followed by a three-point parabolic interpolation of the
maximum.  We use $2\times10^5$ pseudoexperiments at each signal
strength over $Q=2$--$20$.  For the two-dimensional study we take
$\sigma_0=1$ and use 121 points over
$[-7\sigma_0,7\sigma_0]$.  A safeguarded Newton maximization is
restricted to $|m|\le2\sigma_0$ and $|\eta|\le0.9$.  This study uses
$10^6$ pseudoexperiments per signal strength for
$Q=5,6,\ldots,10,12,15,20,25,30$, with the asymptotic coefficient fits
restricted to $Q\ge8$.  Convergence failures, boundary solutions, and
non-negative fitted Hessians were monitored as numerical diagnostics;
they are negligible in the high-$Q$ regime used for the asymptotic
comparisons.

In both examples the metric computed from the discretized templates
agrees with the continuum result: $G\simeq0.5/\sigma^2$ for the
one-dimensional shift and
\begin{equation}
  G_{mm}=G_{\eta\eta}\simeq0.5,
  \qquad
  G_{m\eta}\simeq0
\end{equation}
for the location--scale model.

\subsection{Leading dimension-dependent behavior}

The simulations first test the universal tangent-space predictions,
which depend only on the number of profiled template coordinates.
Figure~\ref{fig:local-validation}(a) shows the scaled amplitude bias
$Q\,\E[\widehat Q-Q]$.  The $D=1$ and $D=2$ results approach
$1/2$ and $1$, respectively, as predicted by
Eq.~\eqref{eq:bias-leading}.  In the one-dimensional model the fitted
location resolution likewise approaches
$\sqrt2\,\sigma/Q$.

At the distribution level, the profile gain approaches the predicted
$\chi^2_D$ law in both examples.
Figure~\ref{fig:local-validation}(b) shows the $D=1$ and $D=2$
distributions at $Q=20$ together with the corresponding
$\chi^2_1$ and $\chi^2_2$ densities.  The whitened localization
coordinates simultaneously approach independent unit Gaussians and
$r_{\rm loc}^2$ approaches $\chi_D^2$.

\subsection{Model-dependent finite-$Q$ corrections}

The same pseudoexperiments also test the first non-universal curvature
corrections.  These comparisons go beyond verifying the
$D$-dimensional tangent limit: they probe higher derivatives of the
specific Gaussian template manifolds.

For the two-dimensional amplitude bias, a useful control-variate
estimator is
\begin{equation}
  B_{\rm CV}
  =
  \widehat Q-Z_{\rm fix},
  \qquad
  Z_{\rm fix}=Z(\thetao).
  \label{eq:control-variate}
\end{equation}
Since $\E[Z_{\rm fix}]=Q$ exactly,
\begin{equation}
  \E[B_{\rm CV}]
  =
  \E[\widehat Q-Q],
\end{equation}
but the dominant radial noise cancels event by event.

Fitting the high-$Q$ simulations to
\begin{equation}
  Q\,\E[B_{\rm CV}]-1
  =
  \frac{b_3}{Q^2}
  +\frac{b_5}{Q^4}
\end{equation}
gives
\begin{equation}
  b_3=2.263\pm0.159,
\end{equation}
consistent with the analytic value $9/4=2.25$.  Similarly,
\begin{align}
  \E[\widehat\eta]
  &=
  -\frac{e_2}{Q^2}
  +O(Q^{-4}),
  &
  e_2
  &=
  0.988\pm0.016,
  \\
  \E[\Delta q]-2
  &=
  \frac{c_2}{Q^2}
  +O(Q^{-4}),
  &
  c_2
  &=
  4.52\pm0.32,
\end{align}
consistent with the predictions $e_2=1$ and $c_2=9/2$.
Figure~\ref{fig:local-validation}(a) also shows the corresponding
one-dimensional correction $3/(8Q^3)$, while
Fig.~\ref{fig:local-validation}(c) displays the width bias
$Q^2\E[\widehat\eta]\to-1$.

The localization observables converge more slowly to their first
asymptotic corrections than the amplitude-bias and profile-gain
observables.  Figure~\ref{fig:local-validation}(d) compares the
measured variances of the whitened coordinates with
\begin{equation}
  \Var(y_m)-1=\frac{6}{Q^2}+O(Q^{-4}),
  \qquad
  \Var(y_\eta)-1=\frac{11}{Q^2}+O(Q^{-4}).
\end{equation}
The simulations approach these asymptotic slopes at the largest signal
strengths but show substantial positive departures at moderate $Q$,
especially in the width direction.  These deviations do not indicate a
failure of the local Gaussian limit: the whitened coordinates still
approach independent unit Gaussians, and $r_{\rm loc}^2$ still approaches
$\chi^2_2$ as $Q\to\infty$.  Rather, they show that the omitted
$O(Q^{-4})$ and higher-order terms are numerically important for
localization observables over the finite-$Q$ range accessible in the
pseudoexperiments.  This contrasts with the amplitude bias and profile
gain, for which the first subleading corrections already provide a
much more accurate description.

\begin{figure}[ht]
  \centering
  \includegraphics[width=\textwidth]{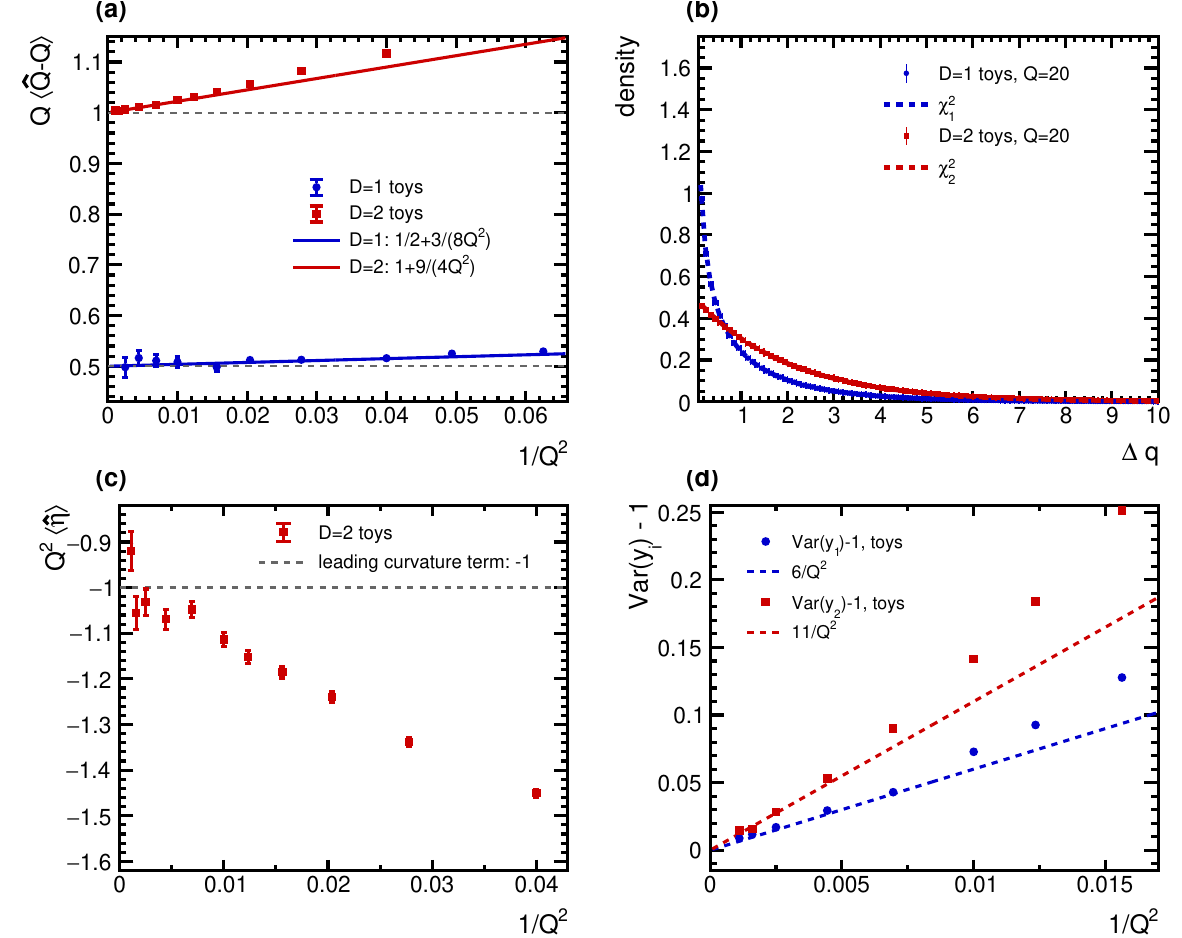}
  \caption{
  Numerical validation of the local profiling theory in the one- and
  two-dimensional Gaussian template models.
  \textbf{(a)} Scaled amplitude bias
  $Q\,\E[\widehat Q-Q]$ as a function of $1/Q^2$.
  The $D=1$ and $D=2$ simulations approach the universal leading limits
  $1/2$ and $1$, respectively, while the first model-dependent
  curvature corrections are described by
  $1/2+3/(8Q^2)$ and $1+9/(4Q^2)$.
  \textbf{(b)} Distribution of the profile gain
  $\Delta q=q_{\rm prof}-q_{\rm fixed}$ at $Q=20$.
  The one- and two-dimensional toy distributions are compared with the
  asymptotic $\chi^2_1$ and $\chi^2_2$ laws.
  \textbf{(c)} Width-coordinate bias in the two-dimensional
  location--scale model.
  The quantity $Q^2\,\E[\widehat\eta]$ approaches the leading curvature
  prediction $-1$ at large $Q$.
  \textbf{(d)} Finite-$Q$ corrections to the variances of the whitened
  localization coordinates.
  The leading predictions
  $\Var(y_m)-1=6/Q^2+O(Q^{-4})$ and
  $\Var(y_\eta)-1=11/Q^2+O(Q^{-4})$
  capture the asymptotic slopes, while the visible departures at
  moderate $Q$ show that higher-order localization corrections are
  numerically important.
  }
  \label{fig:local-validation}
\end{figure}
\clearpage

\section{From local profiling to global competition}
\label{sec:crossover}

The strong-signal expansion cannot be continued to the null simply by
setting $Q=0$: Eq.~\eqref{eq:theta-hat-leading} is singular in that
limit because there is no distinguished true template point around
which to expand.  This is precisely the cone-apex singularity discussed
in Sec.~\ref{sec:model}.  The correct crossover variable is therefore
not a perturbative correction to the local coordinate estimator, but
competition between distinct extrema of the matched-filter field.

\subsection{An exact decomposition}

In the one-dimensional Gaussian-shift example, let $M_S$ denote the
local maximum belonging to the basin of attraction of the injected
signal position $m_0$: starting at $m_0$, one follows the field uphill
to the associated local maximum.  Let $M_R$ denote the largest
\emph{other} competing maximum in the search interval.  The global
maximum is then exactly
\begin{equation}
  M=\max(M_S,M_R).
  \label{eq:global-max-decomp}
\end{equation}
Consequently,
\begin{equation}
  M-Q
  =
  (M_S-Q)
  +
  (M_R-M_S)_+,
  \label{eq:exact-bias-decomp}
\end{equation}
where $(x)_+=\max(x,0)$.  Taking expectations gives
\begin{equation}
  B_{\rm global}(Q)
  =
  B_{\rm local}(Q)
  +
  B_{\rm remote}(Q),
  \label{eq:bias-local-remote}
\end{equation}
with
\begin{equation}
  B_{\rm remote}(Q)
  =
  \E[(M_R-M_S)_+].
  \label{eq:Bremote}
\end{equation}
This decomposition is exact and does not assume independence or a
Poisson process of extrema.

At large $Q$, remote wins are exponentially rare and
$B_{\rm global}\simeq B_{\rm local}$.  At lower $Q$, the second term
activates and introduces the total search range.  In simulations with
$\sigma=1$ and search half-ranges $L=6,10,14$, the local bias at $Q=1$
is approximately $0.527$, while the corresponding global biases are
$0.777$, $0.939$, and $1.053$.  The remote-win probabilities are
approximately $0.354$, $0.498$, and $0.582$, respectively.  By
$Q\simeq5$--$6$, the global and local biases are nearly indistinguishable
on the scale of the simulation.

\subsection{Dependence between signal-associated and remote maxima}

We now test two assumptions that enter when remote competitors are
approximated by an independent background-only field: whether the signal-associated and remote maxima are statistically dependent, and
whether the signal itself changes the marginal distribution of the
remote maximum.

A useful diagnostic separates two approximations that are otherwise
easy to confuse.  
For each signal-containing toy $i$, let $M_{S,i}$ denote the
signal-associated maximum and $M_{R,i}$ the highest competing maximum
outside the signal basin.  We compare
\begin{align}
  M_{\rm actual}
  &=
  \max(M_{S,i},M_{R,i}),
  \\
  M_{\rm shuf}
  &=
  \max(M_{S,i},M_{R,\pi(i)}),
  \\
  M_{\rm null}
  &=
  \max(M_{S,i},M^0_{R,i}),
  \label{eq:actual-shuffled-null}
\end{align}
where $\pi$ is a random permutation of the signal toys and
$M^0_{R,i}$ is the corresponding remote maximum drawn from an
independent background-only toy.  The first construction preserves the
actual event-by-event pairing between the signal-associated and remote
maxima.  The shuffled construction preserves their separate marginal
distributions but removes this dependence.  The final construction
additionally replaces the remote maximum from a signal-containing field
by one drawn from a null field, testing whether the presence of the
signal changes the remote-maxima distribution itself.

Both corrections are small in the regime studied.  For example, at
$L=10$ and $Q=1$ the global biases are
\begin{equation}
  B_{\rm actual}=0.9392,\qquad
  B_{\rm shuf}=0.9472,\qquad
  B_{\rm null}=0.9278.
  \label{eq:dependence-example}
\end{equation}
The three constructions separate two small effects.  First,
$B_{\rm actual}<B_{\rm shuf}$ shows that the event-by-event dependence
between the signal-associated and remote maxima slightly reduces the
global greedy bias.  Second, $B_{\rm shuf}>B_{\rm null}$ shows that
remote maxima drawn from signal-containing toys are slightly more
competitive than those drawn from an independent background-only field.
The two corrections have opposite signs and therefore partially cancel.
Although the ordinary Pearson correlation between $M_S$ and $M_R$ is
small, this does not exclude dependence in the rare tail configurations
that determine whether a remote maximum overtakes the signal-associated
one.

The null-remote approximation becomes particularly informative once winning competitors are classified by distance from the injected signal.  At low and moderate $Q$, most winning competitors are genuinely
remote.  In the rare-win regime at larger $Q$, however, the surviving
competitors are increasingly concentrated near the signal, especially
for shorter searches.  This produces a small signal-tail contamination
of the remote marginal and explains why the highest-$Q$ remote-win
probabilities are not described perfectly by an independent null field.

\subsection{Dimension, curvature, and search volume}

The crossover between the signal-associated maximum and remote
competitors is summarized in Fig.~\ref{fig:crossover}.  It also helps
to separate three levels of the problem.  The dimensionality $D$ of the
template manifold controls the universal leading local effect: it fixes
the number of tangent directions available to the profile fit and hence
the leading $D/(2Q)$ amplitude bias and $\chi^2_D$ profile gain.
Curvature of the template manifold enters only at the next order,
producing model-dependent finite-$Q$ corrections to these local
asymptotic results.
As Fig.~\ref{fig:crossover} shows, increasing the search range has
essentially no effect once the signal is strong enough for the
signal-associated maximum to dominate, but it increases both the
remote contribution and the remote-win probability at smaller $Q$.

The total search volume, boundaries, and correlations among competing
extrema become relevant when the maximization is no longer confined to
the signal basin.  These are truly global ingredients: they control
how often a remote maximum can overtake the signal-associated one and
therefore determine the departure of the global bias from its local
strong-signal limit.  In this regime the remote extrema cannot in
general be represented by a fixed number of independent trials; their
correlations and spatial distribution become part of the problem.

\begin{figure}[t]
  \centering
  \includegraphics[width=\textwidth]{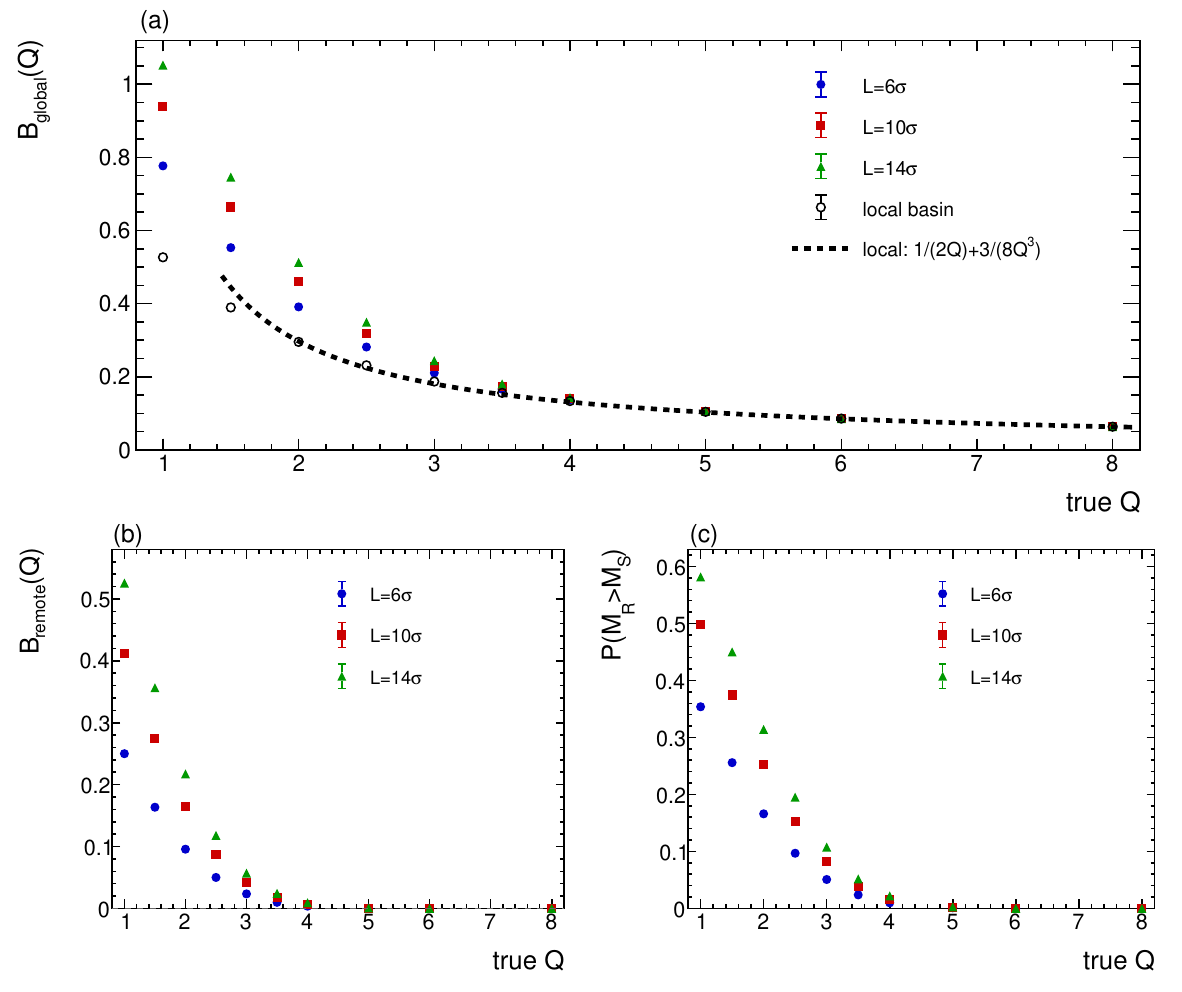}
  \caption{
  Local-to-global crossover in the one-dimensional Gaussian-shift model
  for search half-ranges $L=6\sigma$, $10\sigma$, and $14\sigma$.
  \textbf{(a)} Global greedy bump bias
  $B_{\rm global}(Q)=\E[M-Q]$ as a function of the true signal amplitude
  $Q$.  The signal-associated local-basin bias is also shown, together
  with the strong-signal expansion
  $1/(2Q)+3/(8Q^3)$.
  At large $Q$ the global results become independent of the search range
  and approach the local prediction, whereas at smaller $Q$ the
  increasing probability of a competing remote maximum produces a
  search-volume-dependent excess bias.
  \textbf{(b)} Remote contribution
  $B_{\rm remote}(Q)=\E[(M_R-M_S)_+]$, defined by the exact decomposition
  $B_{\rm global}=B_{\rm local}+B_{\rm remote}$.
  The remote term vanishes in the strong-signal regime and grows with
  the available search range as $Q$ decreases.
  \textbf{(c)} Probability that the highest remote competitor overtakes
  the signal-associated maximum, $\Pr(M_R>M_S)$.
  Together, panels (b) and (c) show how the global look-elsewhere
  behavior emerges continuously from competition between the local
  signal maximum and remote extrema.
  Statistical uncertainties are smaller than the marker size.
  }
  \label{fig:crossover}
\end{figure}
\clearpage

\section{Finite-threshold corrections to the Poisson approximation}
\label{sec:cumulants}

High excursion peaks of a smooth Gaussian field become asymptotically
rare, and the probability of finding no peaks above a high threshold
is well approximated by a Poisson model.  At lower thresholds, however,
correlations among extrema make this Poisson approximation
imperfect.  It is therefore useful to describe the departure from
Poisson behavior directly in terms of the counting process.

\subsection{Factorial cumulants and the void probability}

Let $N_u$ be the number of maxima above level $u$ under the null; the
probability $\Pr_0(N_u=0)$ is commonly called the void probability. Define its
factorial cumulants $\kappa_n(u)$ through
\begin{equation}
  \log \E[(1+t)^{N_u}]
  =
  \sum_{n\ge1}
  \frac{\kappa_n(u)}{n!}\,t^n.
  \label{eq:factorial-cumulant-def}
\end{equation}
In particular,
\begin{align}
  \kappa_1
  &=
  \E[N_u],
  \\
  \kappa_2
  &=
  \E[N_u(N_u-1)]-\E[N_u]^2,
  \label{eq:kappa2-def}
  \\
  \kappa_3
  &=
  \E[N_u(N_u-1)(N_u-2)]
  -3\E[N_u(N_u-1)]\E[N_u]
  +2\E[N_u]^3.
  \label{eq:kappa3-def}
\end{align}
Whenever the series may be evaluated at $t=-1$, the void probability is
\begin{equation}
  \log\Pr_0(N_u=0)
  =
  -\kappa_1
  +\frac{\kappa_2}{2}
  -\frac{\kappa_3}{6}
  +\cdots .
  \label{eq:void-cumulant}
\end{equation}
The Poisson approximation consists of setting all
$\kappa_{n>1}=0$.

Equation~\eqref{eq:void-cumulant} usefully
separates the one-point intensity of maxima from their higher-order
spatial correlations.  The first correction is not another effective
search length; it is the connected two-point structure of the extrema
process.

\subsection{Finite-interval first moment}

For the finite one-dimensional crossover search, the possible
contributors to the global maximum consist of the interior local
maxima together with the two endpoints of the search interval.
To reproduce the local-versus-remote decomposition of
Sec.~\ref{sec:crossover} under the null, we exclude the maximum whose
attraction basin contains the reference point $m_0$.  We denote this
tagged maximum by $M_{\rm tag}^0$ and its survival function by
\begin{equation}
  S_{\rm tag}(u)
  =
  \Pr_0(M_{\rm tag}^0>u).
\end{equation}
Let $N_R(u)$ denote the number of remaining extrema above $u$ after
this tagged maximum is excluded, with the two endpoints retained as
possible contributors.  From this point through the end of this
section, $\kappa_n(u)$ denotes the factorial cumulants of this
remote-extrema counting process.

For the Gaussian matched-filter covariance
\begin{equation}
  \rho(r)
  =
  \exp[-r^2/(4\sigma^2)],
  \label{eq:rho-gaussian-cumulant}
\end{equation}
the one-point Kac--Rice density of interior local maxima of height $x$ is
\begin{equation}
  \nu(x)
  =
  \frac{\phi(x)}{\sigma\sqrt{\pi}}
  \left[
    \frac{\phi(x/\sqrt2)}{\sqrt2}
    +\frac{x}{2}\Phi(x/\sqrt2)
  \right],
  \label{eq:nu-maxima}
\end{equation}
and we write
\begin{equation}
  \nu_{>}(u)
  =
  \int_u^\infty \nu(x)\,dx.
  \label{eq:nu-tail}
\end{equation}
The corresponding first moment for the remote-extrema count is
\begin{equation}
  \kappa_1^{\rm model}(u)
  =
  2L\,\nu_{>}(u)
  +2\overline\Phi(u)
  -S_{\rm tag}(u).
  \label{eq:kappa1-finite}
\end{equation}
The second term is the finite-interval boundary contribution.  It is
numerically important at the moderate thresholds relevant to the
crossover and was absent from our first effective-length approximation.

For $L=10\sigma$, the measured and predicted first moments agree closely;
for example,
\begin{equation}
  \kappa_1^{\rm toy}(1.0)=1.31575,
  \qquad
  \kappa_1^{\rm model}(1.0)=1.32372,
\end{equation}
and at $u=2.0$ the corresponding values are
$0.24217$ and $0.24275$.

\subsection{Size of the non-Poisson correction}

A $5\times10^5$-toy null study makes the factorial hierarchy explicit.
Table~\ref{tab:void-cumulants} gives representative results for
$L=10\sigma$.  The $\kappa_1$-only Poisson approximation misses the void
probability by percent-level amounts in the crossover region.  Including
$\kappa_2$ removes most of the error, and including $\kappa_3$ makes the
result essentially indistinguishable from the toy estimate.

\begin{table}[t]
\centering
\caption{Void probability for the remote-extrema process at
$L=10\sigma$.  Here $F_{\rm toy}=\Pr_0[N_R(u)=0]$,
$F_1=e^{-\kappa_1}$,
$F_2=e^{-\kappa_1+\kappa_2/2}$, and
$F_3=e^{-\kappa_1+\kappa_2/2-\kappa_3/6}$.}
\label{tab:void-cumulants}
\begin{tabular}{ccccc}
\toprule
$u$ & $F_{\rm toy}$ & $F_1$ & $F_2$ & $F_3$ \\
\midrule
1.0 & 0.25336 & 0.26827 & 0.25223 & 0.25349 \\
1.5 & 0.54415 & 0.53249 & 0.54118 & 0.54411 \\
2.0 & 0.79724 & 0.78492 & 0.79656 & 0.79728 \\
2.5 & 0.93388 & 0.92895 & 0.93385 & 0.93388 \\
3.0 & 0.98368 & 0.98246 & 0.98367 & 0.98368 \\
\bottomrule
\end{tabular}
\end{table}

The sign of $\kappa_2$ for this finite, tagged extrema process is not
fixed:
it is negative at lower threshold and becomes mildly positive in
part of the intermediate range.  It is therefore more accurate to speak
of \emph{non-Poisson extremal correlations} than simply of repulsion.
The clean interior-maxima process considered next removes the tagged and
boundary conventions and exposes the underlying two-point random-field
geometry directly.

\begin{figure}[t]
  \centering
  \includegraphics[width=\textwidth]{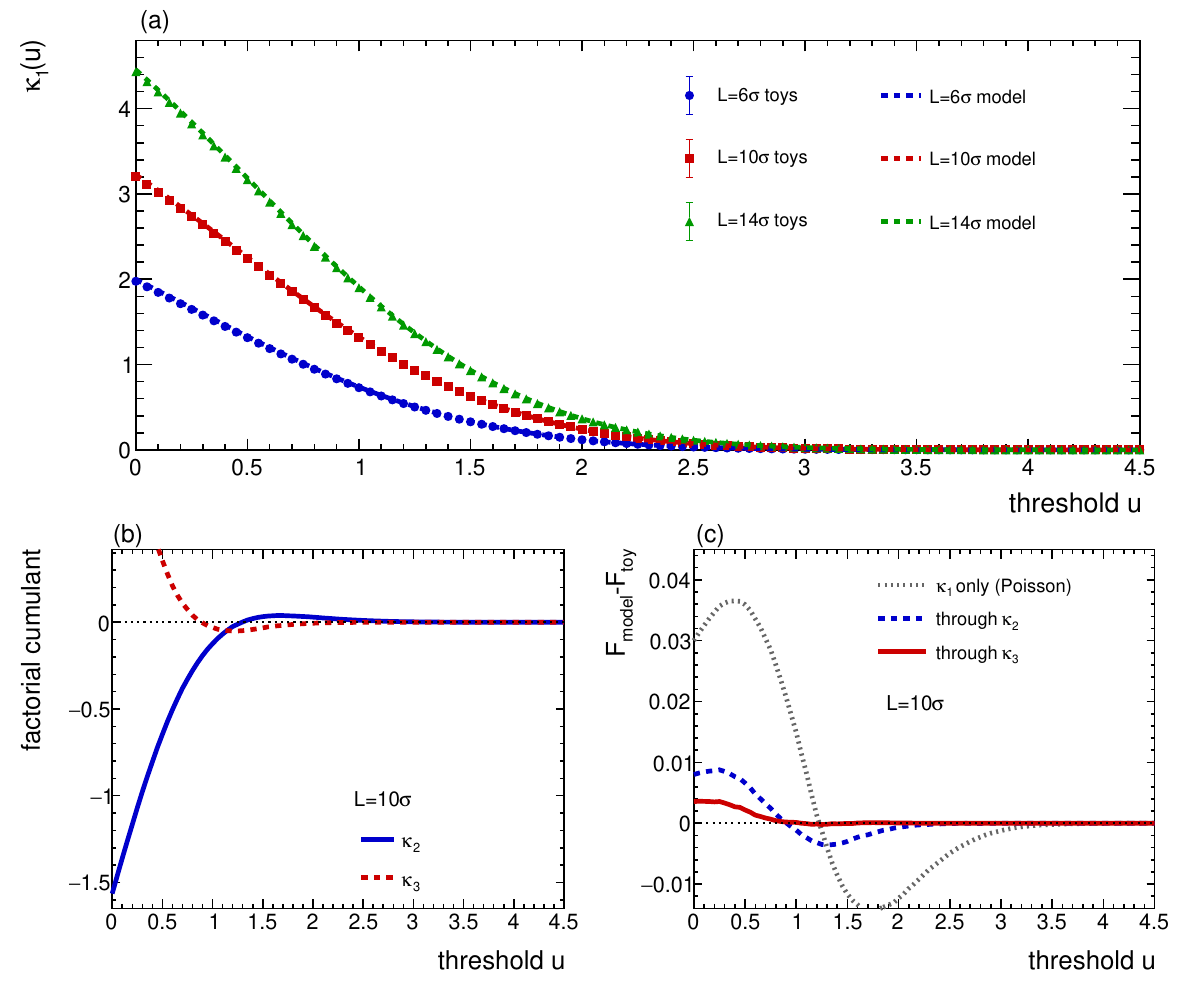}
  \caption{
  Finite-threshold factorial-cumulant corrections for the remote-extrema
  process.
  \textbf{(a)} First factorial cumulant
  $\kappa_1(u)=\E[N_R(u)]$ for search half-ranges
  $L=6\sigma$, $10\sigma$, and $14\sigma$.
  Toy estimates are compared with the one-point model
  $
    \kappa_1^{\rm model}(u)
    =
    2L\,\nu_{>}(u)
    +2\overline{\Phi}(u)
    -S_{\rm tag}(u),
  $
  which includes the interior Kac--Rice contribution, the two
  finite-interval endpoints, and subtraction of the tagged maximum.
  \textbf{(b)} Second and third factorial cumulants,
  $\kappa_2(u)$ and $\kappa_3(u)$, for $L=10\sigma$.
  Their nonzero values demonstrate that the finite-threshold extrema
  process is not Poisson; in particular, the change of sign of
  $\kappa_2$ shows that the correlations cannot be characterized simply
  as repulsive.
  \textbf{(c)} Residual
  $F_{\rm model}(u)-F_{\rm toy}(u)$ of the void probability
  $F(u)=\Pr_0[N_R(u)=0]$ for successive truncations of the factorial-cumulant
  expansion.
  The Poisson approximation, which retains only $\kappa_1$, shows a
  visible finite-threshold discrepancy.
  Including $\kappa_2$ removes most of this error, while including
  $\kappa_3$ makes the prediction essentially indistinguishable from the
  toy result over the range shown.
  }
  \label{fig:cumulants}
\end{figure}
\clearpage


\section{Two-point correlations of null maxima}
\label{sec:kappa2}

Section~\ref{sec:cumulants} showed that at finite threshold the maxima
of the null field are not distributed as an independent Poisson
process.  The leading departure is described by the second factorial
cumulant $\kappa_2$, which measures whether pairs of maxima occur more
or less frequently than they would for independent extrema.  We now
ask where this correction comes from.

The answer can be obtained directly from the geometry of the smooth
Gaussian field.  We calculate the pair density for two local maxima
above threshold at a specified separation $r$, and then integrate the
result over the search interval.  This gives an analytic prediction
for $\kappa_2$, which can be compared directly with the
pseudoexperiments of Sec.~\ref{sec:cumulants}.  The same calculation
also reveals a strong short-distance suppression: two distinct maxima
cannot occur arbitrarily close together in a smooth field.

For this calculation it is convenient to consider only interior local
maxima.  We therefore temporarily remove the two finite-interval
endpoints and the tagged-basin convention used in
Sec.~\ref{sec:cumulants}, leaving the stationary process to which the
standard two-point Kac--Rice formula applies directly
\cite{AzaisWschebor2009,AdlerTaylor2007}.

For this one-dimensional calculation, write
$Z(t)\equiv Z_0(t)$ for the stationary null matched-filter field.
Then
\begin{equation}
  N_u^{\rm int}
  =
  \#\{
    t\in(-L,L):
    Z'(t)=0,\;
    Z''(t)<0,\;
    Z(t)>u
  \},
  \label{eq:Nint}
\end{equation}
so that $N_u^{\rm int}$ simply
counts the interior peaks that exceed the threshold $u$.  Writing
\begin{equation}
  T=2L,
  \qquad
  \lambda(u)=\nu_{>}(u),
\end{equation}
where $\lambda(u)$ is the mean number of such maxima per unit length,
we have
\begin{equation}
  \E[N_u^{\rm int}]
  =
  T\lambda(u).
  \label{eq:kappa1-interior}
\end{equation}

\subsection{The two-point intensity}

Let $R_2(r;u)$ denote the density of pairs of interior local maxima
above $u$ whose separation is $r$.  If the maxima were independent,
this pair density would simply be $\lambda(u)^2$.  The difference
\begin{equation}
  R_2(r;u)-\lambda(u)^2
\end{equation}
therefore measures the excess or deficit of pairs at separation $r$.
Stationarity gives
\begin{equation}
  \E[ 
    N_u^{\rm int}
    (N_u^{\rm int}-1)
  ]
  =
  2\int_0^T
  (T-r)R_2(r;u)\,dr,
  \label{eq:factorial-second-KR}
\end{equation}
and hence
\begin{equation}
  \kappa_2^{\rm int}(u)
  =
  2\int_0^T
  (T-r)
  \left[
    R_2(r;u)-\lambda(u)^2
  \right]dr.
  \label{eq:kappa2-KR-main}
\end{equation}
Thus the second factorial cumulant is the integrated departure from
independent pair counting.

It is useful to express the same information through the dimensionless
pair-correlation function
\begin{equation}
  g_2(r;u)
  =
  \frac{R_2(r;u)}{\lambda(u)^2}.
  \label{eq:g2-def}
\end{equation}
Here $g_2=1$ corresponds to independent maxima, $g_2<1$ to a
suppression of pairs at that separation, and $g_2>1$ to an enhancement.
Equation~\eqref{eq:kappa2-KR-main} is therefore the weighted integral
of the connected pair correlation $g_2-1$.

To have local maxima at two points separated by $r$, three elementary
conditions must hold at both points: the field value must exceed $u$,
the slope must vanish, and the curvature must be negative.  Define
\begin{equation}
  \bm S=(Z'(0),Z'(r)),
  \qquad
  \bm H=(-Z''(0),-Z''(r)),
\end{equation}
so that positive components of $\bm H$ correspond to downward
curvature.  For a scalar curvature value $h$, define
$h_+=\max(h,0)$.  Let $p_{\bm S}(\bm 0)$ denote the
joint probability density of the two slopes evaluated at
$\bm S=\bm 0$.  The two-point Kac--Rice formula then combines
this zero-slope density with the conditional probability that both
points lie above threshold and have the required curvature:
\begin{equation}
  R_2(r;u)
  =
  p_{\bm S}(\bm 0)\,
  \E\!\left[
    H_{0,+}H_{r,+}\,
    \mathbf 1_{\{Z(0)>u,Z(r)>u\}}
    \mid \bm S=\bm 0
  \right].
  \label{eq:R2-KR}
\end{equation}
All covariances entering Eq.~\eqref{eq:R2-KR} follow by differentiating
$\rho(r)$.  For example,
\begin{align}
  \rho'(r)
  &=
  -\frac{r}{2\sigma^2}\rho(r),
  \\
  \rho''(r)
  &=
  \left(
    \frac{r^2}{4\sigma^4}
    -\frac{1}{2\sigma^2}
  \right)\rho(r),
\end{align}
with analogous expressions for the third and fourth derivatives.

\subsection{Principal coordinates}

At small separation the two field values $Z(0)$ and $Z(r)$ become nearly
identical, making a direct numerical integration poorly conditioned.
The natural variables are therefore the symmetric and antisymmetric
height combinations
\begin{equation}
  V=\frac{Z(0)+Z(r)}{\sqrt2},
  \qquad
  W=\frac{Z(0)-Z(r)}{\sqrt2},
  \label{eq:VW}
\end{equation}
whose conditional variances, after imposing $\bm S=\bm 0$, are denoted
\begin{equation}
  \Lambda_+=\Var(V\mid\bm S=\bm 0),
  \qquad
  \Lambda_-=\Var(W\mid\bm S=\bm 0).
\end{equation}
In these coordinates the Gaussian covariance is diagonal and the
numerical integration remains well behaved even as $r\to0$.  The
standardization used for the numerical calculation, and the detailed
small-separation expansion, are given in
Appendix~\ref{app:short-distance}.

\subsection{Comparison with simulation}

Table~\ref{tab:kappa2-theory} compares the converged two-point
Kac--Rice calculation with $5\times10^5$ interior-only null
pseudoexperiments for $L=10\sigma$.  The agreement is at the percent
level or better and all four differences are below two quoted toy
standard errors.

\begin{table}[t]
\centering
\caption{Second factorial cumulant of the clean interior-maxima process
for $L=10\sigma$.  The final column is
$(\kappa_2^{\rm toy}-\kappa_2^{\rm KR})/\sigma_{\rm toy}$.}
\label{tab:kappa2-theory}
\begin{tabular}{cccc}
\toprule
$u$ & $\kappa_2^{\rm toy}$ &
$\kappa_2^{\rm KR}$ & pull \\
\midrule
0.5 & $-1.173308\pm0.002739$ & $-1.177371$ & $+1.48$ \\
1.0 & $-0.467233\pm0.001961$ & $-0.470478$ & $+1.65$ \\
1.5 & $-0.118681\pm0.001048$ & $-0.120039$ & $+1.30$ \\
2.0 & $-0.018479\pm0.000466$ & $-0.018343$ & $-0.29$ \\
\bottomrule
\end{tabular}
\end{table}

The structure behind these numbers is shown in
Fig.~\ref{fig:kappa2-KR}.  Panel~(a) displays $g_2(r;u)$.  Nearby
maxima are strongly suppressed, so $g_2<1$ at small $r$; this is
followed by a modest enhancement at intermediate separation, while at
large separation $g_2\to1$ and the maxima become effectively
independent.  The negative values of $\kappa_2^{\rm int}$ in
Table~\ref{tab:kappa2-theory} are therefore the integrated consequence
of this nontrivial pair structure rather than of a fixed number of
independent ``correlation cells.''  Panel~(c) shows that integrating
the calculated pair correlation reproduces the second factorial
cumulant measured in the pseudoexperiments over all three search
ranges.

\subsection{Quartic exclusion of nearby maxima}

There is a simple geometric reason for the strong short-distance
suppression.  Two distinct maxima separated by a very small distance
cannot be chosen independently.  Because the field is smooth, their
heights become almost equal, both slopes must vanish, and the field
must bend downward at both locations while accommodating the turning
structure between them.  The allowed configurations therefore occupy
a rapidly shrinking region of the joint Gaussian phase space as
$r\to0$.  The principal-coordinate expansion makes this statement
quantitative.

Writing $x=r/\sigma$, the conditional variance of the antisymmetric
height mode vanishes rapidly,
\begin{equation}
  \Lambda_-
  =
  \frac{x^6}{384}
  +O(x^8),
  \label{eq:lambda-minus-leading}
\end{equation}
whereas the symmetric variance remains finite,
$\Lambda_+=4/3+O(x^2)$.  Requiring both points to be maxima further
restricts the antisymmetric mode, while the two positive curvatures
are themselves of order $x^2$.  The apparent $1/x$ singularity in the
zero-slope density is cancelled by the shrinking available phase
space.  Combining these factors gives
\begin{equation}
  R_2(r;u)
  =
  \frac{A(u)}{\sigma^2}
  \left(\frac{r}{\sigma}\right)^4
  +o(r^4),
  \label{eq:R2-quartic}
\end{equation}
or, equivalently,
\begin{equation}
  \boxed{
  g_2(r;u)
  \sim
  C(u)
  \left(\frac{r}{\sigma}\right)^4.
  }
  \label{eq:g2-quartic}
\end{equation}
As shown in Fig.~\ref{fig:kappa2-KR}(b), dividing the numerical pair
correlation by $(r/\sigma)^4$ produces a finite limit as $r\to0$,
confirming this quartic exclusion law.

The coefficient $C(u)$ is given in closed form in
Appendix~\ref{app:short-distance}.  For
$u=0.5,1,1.5,2$ it is approximately
\begin{equation}
  C(u)
  =
  0.01672,\;
  0.02399,\;
  0.03884,\;
  0.07118,
  \label{eq:C-values}
\end{equation}
respectively.  For example, at $u=1$ and $r=0.1\sigma$,
Eq.~\eqref{eq:g2-quartic} predicts $g_2=2.399\times10^{-6}$, compared
with $2.40\times10^{-6}$ from the full two-point calculation.

\begin{figure}[t]
  \centering
  \includegraphics[width=\textwidth]{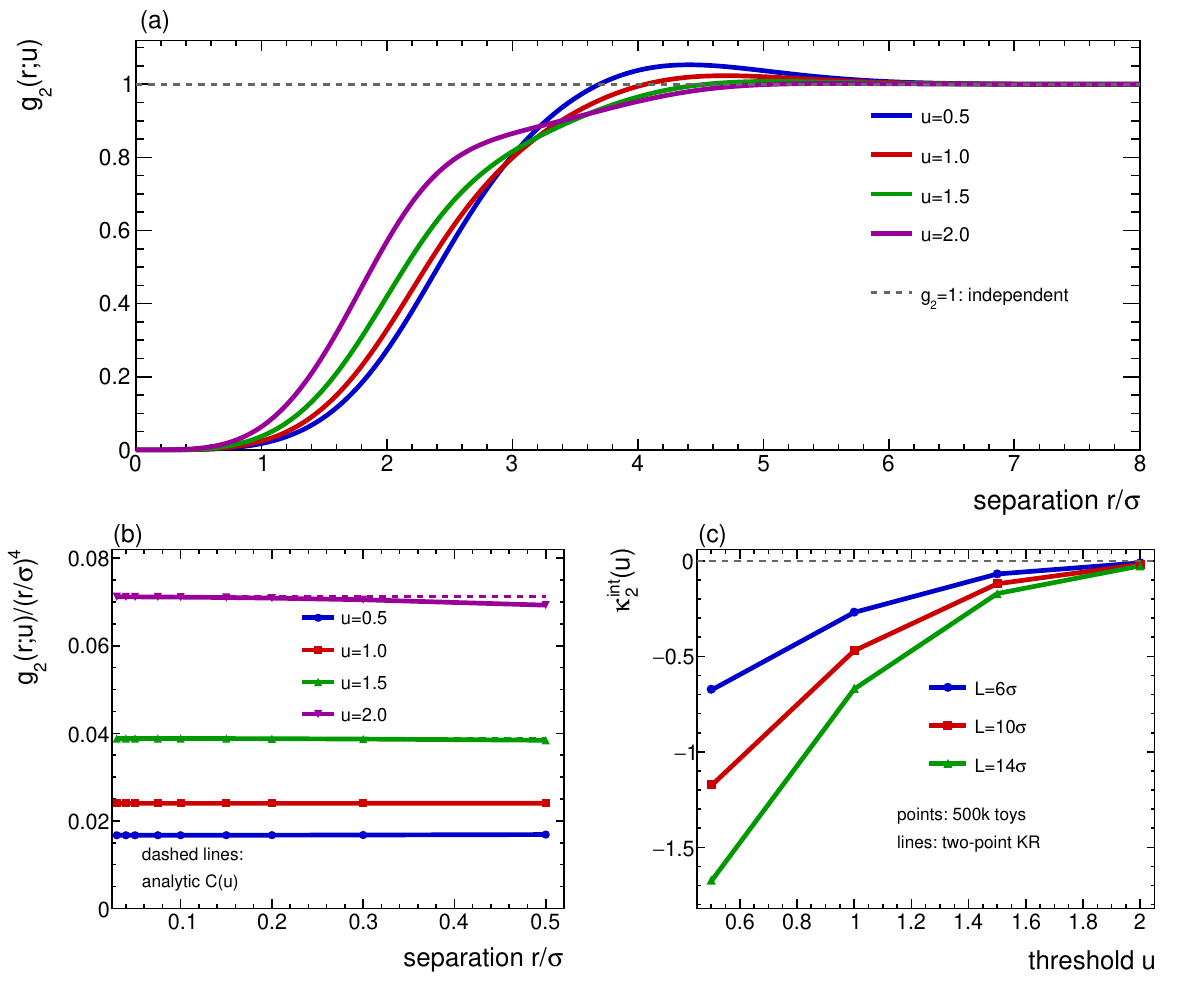}
  \caption{
  Two-point structure of the interior-maxima process in the
  one-dimensional Gaussian matched-filter field.
  \textbf{(a)} Pair correlation
  $g_2(r;u)=R_2(r;u)/\lambda(u)^2$ for thresholds
  $u=0.5,1,1.5,2$.
  Nearby maxima are strongly suppressed, the correlation develops a
  modest intermediate-separation enhancement, and
  $g_2(r;u)\to1$ at large separation, corresponding to asymptotically
  independent extrema.
  \textbf{(b)} Short-distance behavior of the same pair correlation,
  displayed as $g_2(r;u)/(r/\sigma)^4$.
  The numerical two-point Kac--Rice calculation approaches the analytic
  constants $C(u)$, shown by the dashed lines, confirming the quartic
  exclusion law
  $g_2(r;u)\sim C(u)(r/\sigma)^4$ as $r\to0$.
  \textbf{(c)} Second factorial cumulant of the clean interior-maxima
  process for search half-ranges $L=6\sigma$, $10\sigma$, and
  $14\sigma$.
  Points show results from $5\times10^5$ pseudoexperiments and lines show
  the two-point Kac--Rice prediction
  $
    \kappa_2^{\rm int}(u)
    =
    2\int_0^T (T-r)
    \left[R_2(r;u)-\lambda(u)^2\right]dr,
    \qquad T=2L.
  $
  The calculation reproduces both the threshold dependence and the
  search-range dependence of the simulated second factorial cumulant at
  the percent level or better; the small residual differences are
  consistent with the finite lattice spacing of the pseudoexperiments
  relative to the continuum Kac--Rice calculation.
  }
  \label{fig:kappa2-KR}
\end{figure}
\clearpage

\section{Conclusions}
\label{sec:conclusions}

We have placed the greedy bump bias and the look-elsewhere effect in a
common matched-filter geometry.  Profiling over $D$ smooth signal
coordinates exposes $D$ Gaussian tangent directions.  Under a strong
signal, their squared components yield an asymptotic $\chi_D^2$ profile
gain and a leading positive amplitude shift $D/(2Q)$.  Under the null,
the same template metric controls the conditional curvature of a high
excursion, producing the $u^D$ factor in the leading trials factor.  This
gives the asymptotic differential bridge
\begin{equation}
  \boxed{
  \E_Q[\widehat Q-Q]
  =
  \left.
  \frac12\frac{d}{du}\log{\rm TF}_0(u)
  \right|_{u=Q}
  +o(Q^{-1}),
  }
\end{equation}
together with the corresponding likelihood-gain relation.  The
one-dimensional floating-location and two-dimensional
location--scale examples verify the universal leading laws and exhibit
model-dependent higher-order curvature corrections.

The global problem separates naturally into the signal-associated local
maximum and the best competing maximum.  Their maximum gives an exact
decomposition of the global bias into local profiling and remote-peak
competition.  This makes clear why the strong-signal bias is insensitive
to total search volume at leading order, while the weak- and
moderate-signal regimes recover the familiar global multiplicity.

At finite threshold the extrema process contains information beyond its
mean density.  Writing the null void probability in factorial cumulants
shows directly where the Poisson approximation breaks down.  In the
Gaussian one-dimensional example, the second factorial cumulant captures
the leading correction and the third makes the void probability nearly
exact over the crossover range.  For the clean interior-maxima process,
we calculate $\kappa_2$ from the two-point Kac--Rice density and find
quantitative agreement with simulation.  The same calculation yields a
quartic exclusion law
\begin{equation}
  g_2(r;u)
  \sim
  C(u)(r/\sigma)^4
\end{equation}
for nearby maxima.

The resulting hierarchy is therefore more precise than a single
effective-trials-factor picture.  Local dimension controls the universal
strong-signal bias; curvature controls its subleading corrections;
one-point Kac--Rice controls the leading null multiplicity; and
two-point extremal correlations control the first finite-threshold
departure from the Poisson limit.  In this sense the greedy-bump and
look-elsewhere problems are not identical phenomena, but complementary
limits of the same profiling and random-field geometry.

\subsection{Scope and open directions}
\label{sec:outlook}

The analysis above establishes the leading local-to-global connection
for the Gaussian matched-filter model, but several extensions remain.
The exact decomposition of the global bias into signal-associated and
remote contributions is general, whereas the factorization of those
two extrema is only approximate.  The residual dependence is small in
the examples studied, but the rare high-$Q$ remote wins reveal a
signal-tail contamination of nearby competitors.  A fully
non-centered two-point treatment could quantify this correction.
Likewise, the empirical third factorial cumulant removes most of the
small residual left after the $\kappa_2$ correction; a three-point
Kac--Rice calculation would provide its analytic counterpart.

Other extensions concern the local geometry itself.  Physical
constraints, such as a non-negative width or signal strength, can
truncate tangent directions, while nearly degenerate signal
parameters can reduce the effective rank of $G$.  This suggests
generalizations in which the nominal dimension $D$ is replaced by a
boundary- or rank-dependent effective dimension.  Beyond leading
order, the $D=1$ and $D=2$ examples also show that the amplitude and
localization corrections depend on higher derivatives of the template
manifold.  Expressing these terms in coordinate-invariant geometric
form would clarify which features remain universal beyond the tangent
approximation.

Finally, the Gaussian-noise model used here was chosen to isolate the
geometry cleanly.  Realistic resonance searches are more often based
on Poisson or extended likelihoods.  Establishing precisely which of
the tangent-space, local-bias, and finite-threshold results survive
under local asymptotic normality is therefore an important next step.

The extensions described above would test how far the common geometric picture developed here survives beyond the idealized Gaussian setting.

\section*{Acknowledgments}

The author acknowledges the assistance from OpenAI's ChatGPT (GPT-5.6 Sol) in the development and checking of analytic derivations, numerical studies and code, literature analysis, and preparation of the
manuscript. The scientific direction, interpretation of the results, and final responsibility for the content remain with the author.

\clearpage
\appendix
\titleformat{\section}
  {\normalfont\Large\bfseries}
  {Appendix \thesection:}
  {0.5em}
  {}
\section{Higher-order local expansions}
\label{app:local-higher-order}

This appendix derives the model-specific higher-order coefficients
quoted in Sec.~\ref{sec:examples}.  The calculation is a perturbative
solution of the local stationarity equations in powers of $1/Q$,
followed by Gaussian moment contractions.  No fitted numerical
coefficient enters the derivation.

\subsection{General perturbative expansion}
\label{app:local-general}

Choose local coordinates such that the true template is at
$\bm\theta=\bm0$, and write
\begin{equation}
  Z_Q(\bm\theta)
  =
  Q\,\rho(\bm\theta)+W(\bm\theta),
  \qquad
  \rho(\bm0)=1,
  \qquad
  \nabla\rho(\bm0)=\bm0,
\end{equation}
with
\begin{equation}
  -\partial_i\partial_j\rho(\bm0)=G_{ij}.
\end{equation}
At the true point define the noise derivatives
\begin{equation}
  g_i=\partial_iW,
  \qquad
  H_{ij}=\partial_i\partial_jW,
  \qquad
  T_{ijk}=\partial_i\partial_j\partial_kW,
\end{equation}
and the deterministic overlap tensors
\begin{equation}
  R_{ijk}
  =
  \partial_i\partial_j\partial_k\rho,
  \qquad
  R_{ijkl}
  =
  \partial_i\partial_j\partial_k\partial_l\rho,
\end{equation}
all evaluated at $\bm\theta=\bm0$.

Since $\widehat{\bm\theta}=O_{\rm p}(Q^{-1})$, write
\begin{equation}
  \widehat{\bm\theta}
  =
  \frac{\bm a}{Q}
  +\frac{\bm b}{Q^2}
  +\frac{\bm c}{Q^3}
  +O_{\rm p}(Q^{-4}).
  \label{eq:app-theta-series}
\end{equation}
Expanding
$\nabla Z_Q(\widehat{\bm\theta})=\bm0$ and matching successive powers
of $1/Q$ gives
\begin{align}
  \bm a
  &=
  G^{-1}\bm g,
  \label{eq:app-a}
  \\
  \bm b
  &=
  G^{-1}
  \left\{
    H\bm a
    +\frac12 R_3[\bm a,\bm a]
  \right\},
  \label{eq:app-b}
  \\
  \bm c
  &=
  G^{-1}
  \left\{
    H\bm b
    +\frac12 T[\bm a,\bm a]
    +R_3[\bm a,\bm b]
    +\frac16R_4[\bm a,\bm a,\bm a]
  \right\}.
  \label{eq:app-c}
\end{align}
Here, for example,
$[R_3[\bm u,\bm v]]_i=R_{ijk}u_jv_k$ and
$[T[\bm u,\bm v]]_i=T_{ijk}u_jv_k$; repeated indices are summed.

It is also useful to expand the increase of the fitted peak above its
value at the true fixed template,
\begin{equation}
  \delta
  \equiv
  \widehat Q-Z_{\rm fix}
  =
  \frac{\delta_1}{Q}
  +\frac{\delta_2}{Q^2}
  +\frac{\delta_3}{Q^3}
  +O_{\rm p}(Q^{-4}).
  \label{eq:app-delta-series}
\end{equation}
Using $\bm g=G\bm a$, the first three coefficients reduce to
\begin{align}
  \delta_1
  &=
  \frac12\bm a^TG\bm a,
  \label{eq:app-delta1}
  \\
  \delta_2
  &=
  \frac12\bm a^TH\bm a
  +\frac16R_3[\bm a,\bm a,\bm a],
  \label{eq:app-delta2}
  \\
  \delta_3
  &=
  -\frac12\bm b^TG\bm b
  +\bm a^TH\bm b
  +\frac12R_3[\bm a,\bm a,\bm b]
  \notag\\
  &\hspace{1.5cm}
  +\frac16T[\bm a,\bm a,\bm a]
  +\frac1{24}R_4[\bm a,\bm a,\bm a,\bm a].
  \label{eq:app-delta3}
\end{align}
Equations~\eqref{eq:app-a}--\eqref{eq:app-delta3} are the common
starting point for the two examples below.

\subsection{One-dimensional Gaussian location model}
\label{app:d1-higher-order}

In the metric-whitened coordinate
\begin{equation}
  t=\frac{m-m_0}{\sqrt2\,\sigma},
\end{equation}
the deterministic overlap and the covariance of the noise field are
both
\begin{equation}
  \rho(t)=e^{-t^2/2},
  \qquad
  \E[W(t)W(t')]
  =
  e^{-(t-t')^2/2}.
  \label{eq:app-d1-kernel}
\end{equation}
Define the derivatives at the true point
\begin{equation}
  X=W(0),\quad
  A=W'(0),\quad
  B=W''(0),\quad
  C=W^{(3)}(0),\quad
  D_4=W^{(4)}(0),\quad
  E_5=W^{(5)}(0).
\end{equation}
Their covariances follow by differentiating
Eq.~\eqref{eq:app-d1-kernel}.  The nonzero pairings needed below are
\begin{align}
  \E[A^2]&=1,
  &
  \E[AC]&=-3,
  &
  \E[AE_5]&=15,
  \notag\\
  \E[B^2]&=3,
  &
  \E[BD_4]&=-15,
  &
  \E[C^2]&=15,
  \label{eq:app-d1-pairings}
  \\
  \E[XB]&=-1,
  &
  \E[XD_4]&=3.
  \notag
\end{align}
All required higher moments then follow from Isserlis' theorem.

The stationarity equation is
\begin{equation}
  0
  =
  -Q\widehat t\,e^{-\widehat t^2/2}
  +A+B\widehat t
  +\frac12C\widehat t^2
  +\frac16D_4\widehat t^3+\cdots.
  \label{eq:app-d1-stationarity}
\end{equation}
Solving iteratively gives, through the order needed to expose the first
curvature correction,
\begin{equation}
  \widehat t
  =
  \frac{A}{Q}
  +\frac{AB}{Q^2}
  +\frac{A(A^2+AC+2B^2)}{2Q^3}
  +O_{\rm p}(Q^{-4}).
  \label{eq:app-d1-that}
\end{equation}
Substitution into the peak height gives
\begin{align}
  \delta
  &=
  \frac{A^2}{2Q}
  +\frac{A^2B}{2Q^2}
  \notag\\
  &\quad
  +\frac{1}{Q^3}
  \left(
    \frac{A^4}{8}
    +\frac{A^3C}{6}
    +\frac{A^2B^2}{2}
  \right)
  +O_{\rm p}(Q^{-4}).
  \label{eq:app-d1-delta3-explicit}
\end{align}
For example,
\begin{align}
  \E[\delta_3]
  &=
  \frac18\E[A^4]
  +\frac16\E[A^3C]
  +\frac12\E[A^2B^2]
  \notag\\
  &=
  \frac18(3)
  +\frac16(-9)
  +\frac12(3)
  =
  \frac38,
  \label{eq:app-d1-three-eighths}
\end{align}
which is the $1/Q^3$ coefficient quoted in
Eq.~\eqref{eq:d1-bias-nlo}.

Continuing the same recursion by two further orders gives the
coefficient of $Q^{-5}$ in $\delta$ as
\begin{align}
  \delta_5
  ={}&
  \frac{5A^6}{48}
  +\frac{A^5C}{4}
  +\frac{A^5E_5}{120}
  +\frac{5A^4B^2}{4}
  +\frac{A^4BD_4}{6}
  \notag\\
  &+
  \frac{A^4C^2}{8}
  +A^3B^2C
  +\frac{A^2B^4}{2}.
  \label{eq:app-d1-delta5}
\end{align}
Applying the same pairings gives
\begin{equation}
  \E[\delta_1]=\frac12,
  \qquad
  \E[\delta_2]=0,
  \qquad
  \E[\delta_3]=\frac38,
  \qquad
  \E[\delta_4]=0,
  \qquad
  \E[\delta_5]=\frac{25}{16}.
  \label{eq:app-d1-delta-means}
\end{equation}
Since $\E[Z_{\rm fix}]=Q$, this reproduces
\begin{equation}
  \E[\widehat Q-Q]
  =
  \frac{1}{2Q}
  +\frac{3}{8Q^3}
  +\frac{25}{16Q^5}
  +O(Q^{-7}).
\end{equation}

The same solution for $\widehat t$ yields
\begin{equation}
  \E[\widehat t^2]
  =
  \frac{1}{Q^2}
  +\frac{3}{Q^4}
  +\frac{45}{2Q^6}
  +O(Q^{-8}),
\end{equation}
while
\begin{equation}
  \Delta q
  =
  2Z_{\rm fix}\delta+\delta^2,
  \qquad
  Z_{\rm fix}=Q+X,
\end{equation}
together with the additional pairings involving $X$ in
Eq.~\eqref{eq:app-d1-pairings}, gives
\begin{equation}
  \E[\Delta q]
  =
  1+\frac{1}{2Q^2}
  +\frac{2}{Q^4}
  +O(Q^{-6}).
\end{equation}
Thus all three one-dimensional expansions used in
Sec.~\ref{sec:d1} follow from the same local perturbative solution.

\subsection{Two-dimensional Gaussian location--scale model}
\label{app:d2-higher-order}

For the location--scale example let
$\bm\theta=(m,\eta)$, with $m$ measured in units of $\sigma_0$ as in
Sec.~\ref{sec:d2}.  The covariance kernel of the normalized Gaussian
template field is
\begin{align}
  K\!\left((m,\eta),(m',\eta')\right)
  ={}&
  \left[
    \frac{2e^{\eta+\eta'}}
         {e^{2\eta}+e^{2\eta'}}
  \right]^{1/2}
  \notag\\
  &\times
  \exp\left[
    -\frac{(m-m')^2}
           {2(e^{2\eta}+e^{2\eta'})}
  \right].
  \label{eq:app-d2-kernel}
\end{align}
Setting $(m',\eta')=(0,0)$ recovers the deterministic overlap
Eq.~\eqref{eq:d2-overlap}.  At the true point,
\begin{equation}
  G=\frac12 I_2.
\end{equation}
The nonzero third- and fourth-order deterministic derivatives needed at
this order are
\begin{align}
  R_{mm\eta}
  &=
  \frac12,
  \label{eq:app-d2-R3}
  \\
  R_{mmmm}
  &=
  \frac34,
  &
  R_{mm\eta\eta}
  &=
  \frac14,
  &
  R_{\eta\eta\eta\eta}
  &=
  \frac74,
  \label{eq:app-d2-R4}
\end{align}
together with permutations of the indices.

All noise-derivative covariances are generated directly from
Eq.~\eqref{eq:app-d2-kernel}:
\begin{equation}
  \Cov\!\left(
    \partial^\alpha W,
    \partial^\beta W
  \right)
  =
  \left.
  \partial_{\bm\theta}^{\alpha}
  \partial_{\bm\theta'}^{\beta}
  K(\bm\theta,\bm\theta')
  \right|_{\bm\theta=\bm\theta'=\bm0}.
  \label{eq:app-d2-derivative-cov}
\end{equation}
For example, Eq.~\eqref{eq:app-b} gives
\begin{equation}
  b_\eta
  =
  4g_\eta H_{\eta\eta}
  +4g_mH_{m\eta}
  +2g_m^2.
  \label{eq:app-d2-beta}
\end{equation}
Differentiating the kernel gives
\begin{equation}
  \E[g_m^2]=\frac12,
  \qquad
  \E[g_mH_{m\eta}]=-\frac12,
  \qquad
  \E[g_\eta H_{\eta\eta}]=0,
\end{equation}
and hence
\begin{equation}
  \E[b_\eta]=-1.
\end{equation}
Since the leading coefficient $a_\eta$ has zero mean,
Eq.~\eqref{eq:app-theta-series} therefore gives
\begin{equation}
  \E[\widehat\eta]
  =
  -\frac{1}{Q^2}
  +O(Q^{-4}).
\end{equation}

For the remaining observables it is most compact to quote the Gaussian
contractions generated by
Eqs.~\eqref{eq:app-d2-kernel}--\eqref{eq:app-d2-derivative-cov}.
With $X=W(\bm0)$, the coefficients of the general expansion satisfy
\begin{align}
  \E[a_m^2]
  &=
  \E[a_\eta^2]
  =
  2,
  &
  \E[a_mb_m]
  &=
  \E[a_\eta b_\eta]
  =
  0,
  \notag\\
  \E[b_m]
  &=
  0,
  &
  \E[b_\eta]
  &=
  -1,
  \label{eq:app-d2-contractions-a}
  \\
  \E[b_m^2]
  &=
  12,
  &
  \E[b_\eta^2]
  &=
  23,
  &
  \E[a_mc_m]
  &=
  \E[a_\eta c_\eta]
  =
  0,
  \notag\\
  \E[\delta_1]
  &=
  1,
  &
  \E[\delta_2]
  &=
  0,
  &
  \E[\delta_3]
  &=
  \frac94,
  \label{eq:app-d2-contractions-b}
  \\
  \E[X\delta_2]
  &=
  -1,
  &
  \E[\delta_1^2]
  &=
  2.
  \notag
\end{align}
These are direct Wick contractions of the derivative variables; no
simulation input is used.

The amplitude result follows immediately from
Eq.~\eqref{eq:app-delta-series}:
\begin{equation}
  \E[\widehat Q-Q]
  =
  \frac1Q+\frac{9}{4Q^3}+O(Q^{-5}).
\end{equation}
For the profile gain,
\begin{equation}
  \Delta q
  =
  2(Q+X)\delta+\delta^2,
\end{equation}
so its $Q^{-2}$ coefficient is
\begin{align}
  2\E[\delta_3]
  +2\E[X\delta_2]
  +\E[\delta_1^2]
  &=
  2\left(\frac94\right)+2(-1)+2
  \notag\\
  &=
  \frac92.
\end{align}
Therefore
\begin{equation}
  \E[\Delta q]
  =
  2+\frac{9}{2Q^2}+O(Q^{-4}).
\end{equation}

Finally, because $G^{1/2}=I_2/\sqrt2$,
\begin{equation}
  \bm y_{\rm loc}
  =
  \frac{1}{\sqrt2}
  \left(
    \bm a+\frac{\bm b}{Q}+\frac{\bm c}{Q^2}
  \right)
  +O_{\rm p}(Q^{-3}).
\end{equation}
Using Eqs.~\eqref{eq:app-d2-contractions-a} and
\eqref{eq:app-d2-contractions-b}, including the subtraction of the
nonzero mean in the $\eta$ direction, gives
\begin{align}
  \Var(y_m)
  &=
  1+\frac{6}{Q^2}+O(Q^{-4}),
  \\
  \Var(y_\eta)
  &=
  1+\frac{11}{Q^2}+O(Q^{-4}),
\end{align}
whereas the uncentered localization radius satisfies
\begin{equation}
  \E[r_{\rm loc}^2]
  =
  2+\frac{35}{2Q^2}+O(Q^{-4}).
\end{equation}
This reproduces all of the location--scale coefficients quoted in
Sec.~\ref{sec:d2} and makes explicit which results are fixed by the
leading metric and which first depend on higher derivatives of the
template manifold.

\section{Short-distance principal-coordinate expansion}
\label{app:short-distance}

This appendix records the small-separation calculation underlying
Eq.~\eqref{eq:g2-quartic}.  Write
\begin{equation}
  x=\frac{r}{\sigma},
  \qquad
  \rho(x)=e^{-x^2/4}.
\end{equation}

After conditioning on $\bm S=(Z'(0),Z'(r))=\bm 0$, introduce the height
eigenmodes of Eq.~\eqref{eq:VW}.  The origin of the strong suppression
of the antisymmetric mode can be seen explicitly.  Writing
\begin{equation}
  S_+
  =
  \frac{Z'(0)+Z'(r)}{\sqrt2},
\end{equation}
its variance is
\begin{equation}
  d_+(x)
  \equiv
  \Var(S_+)
  =
  \frac{1}{\sigma^2}
  \left[
    \frac12+
    \left(
      \frac12-\frac{x^2}{4}
    \right)\rho(x)
  \right].
\end{equation}
Since
\begin{equation}
  W=\frac{Z(0)-Z(r)}{\sqrt2},
  \qquad
  \Var(W)=1-\rho(x),
\end{equation}
and
\begin{equation}
  \Cov(W,S_+)
  =
  -\frac{x}{2\sigma}\rho(x),
\end{equation}
conditioning on the two vanishing slopes gives
\begin{equation}
  \Lambda_-
  =
  \Var(W\mid\bm S=\bm 0)
  =
  1-\rho(x)
  -
  \frac{x^2\rho(x)^2}
       {4\sigma^2 d_+(x)}.
  \label{eq:app-lambda-minus-exact}
\end{equation}
Expanding this expression for small $x$, the terms of order $x^2$
and $x^4$ cancel exactly.  Together with the corresponding expansion
of the symmetric mode, one obtains
\begin{align}
  \Lambda_+
  &=
  \frac43
  -\frac{x^2}{18}
  +\frac{x^4}{216}
  -\frac{x^6}{3456}
  +O(x^8),
  \label{eq:app-lambda-plus}
  \\
  \Lambda_-
  &=
  \frac{x^6}{384}
  +\frac{x^8}{3072}
  +O(x^{10}).
  \label{eq:app-lambda-minus}
\end{align}
Thus
\begin{equation}
  \frac{\Lambda_-}{\Lambda_+}
  =
  \frac{x^6}{512}
  +O(x^8).
  \label{eq:app-lambda-ratio}
\end{equation}

For the curvature eigenmodes
\begin{equation}
  H_\pm
  =
  \frac{H_0\pm H_r}{\sqrt2},
\end{equation}
the regressions on the corresponding height modes are
\begin{align}
  \beta_+
  &=
  \frac1{\sigma^2}
  \left[
    \frac{x^2}{16}
    -\frac{x^4}{192}
    +O(x^6)
  \right],
  \\
  \beta_-
  &=
  \frac1{\sigma^2}
  \left[
    \frac{12}{x^2}
    -1
    +\frac{7x^2}{80}
    +O(x^4)
  \right].
\end{align}
After also conditioning on the two heights, the residual curvature
eigenvariances are
\begin{align}
  V_+
  &=
  \frac1{\sigma^4}
  \left[
    \frac{x^4}{48}
    -\frac{x^6}{384}
    +O(x^8)
  \right],
  \\
  V_-
  &=
  \frac1{\sigma^4}
  \left[
    \frac{x^6}{1920}
    -\frac{x^8}{15360}
    +O(x^{10})
  \right].
\end{align}
Hence
\begin{equation}
  \frac{V_-}{V_+}
  =
  \frac{x^2}{40}
  +O(x^4),
  \qquad
  \rho_H
  =
  1-\frac{x^2}{20}
  +O(x^4).
  \label{eq:app-curvature-ratio}
\end{equation}

The derivative covariance eigenvalues behave as
\begin{align}
  d_+
  &=
  \frac1{\sigma^2}
  \left[
    1-\frac{3x^2}{8}
    +O(x^4)
  \right],
  \\
  d_-
  &=
  \frac1{\sigma^2}
  \left[
    \frac{3x^2}{8}
    -\frac{5x^4}{64}
    +O(x^6)
  \right],
\end{align}
so the zero-slope density is
\begin{equation}
  p_{\bm S}(\bm 0)
  \sim
  \frac{\sigma^2\sqrt6}{3\pi x}.
  \label{eq:app-pS}
\end{equation}

Standardize the height modes as
\begin{equation}
  V=\sqrt{\Lambda_+}\,z,
  \qquad
  W=\sqrt{\Lambda_-}\,z_-.
\end{equation}
Generic $z_-=O(1)$ produces opposite curvature shifts of order $x$,
which overwhelm the common $O(x^2)$ curvature fluctuation and prevent
both points from being maxima.  The simultaneous-maxima contribution
therefore comes from
\begin{equation}
  z_-=x\,y,
  \qquad
  y=O(1).
  \label{eq:app-zminus-scaling}
\end{equation}
To leading order,
\begin{align}
  V
  &=
  \frac{2}{\sqrt3}z+O(x^2),
  \\
  W
  &=
  \frac{x^4}{8\sqrt6}y+O(x^6).
\end{align}
If $\xi\sim N(0,1)$ denotes the leading symmetric residual-curvature
fluctuation, the physical curvatures satisfy
\begin{align}
  \frac{\sigma^2H_0}{x^2}
  &=
  \frac{z+2\xi}{8\sqrt6}
  +\frac{\sqrt3}{4}y
  +o(1),
  \\
  \frac{\sigma^2H_r}{x^2}
  &=
  \frac{z+2\xi}{8\sqrt6}
  -\frac{\sqrt3}{4}y
  +o(1).
  \label{eq:app-curvatures}
\end{align}
The $1/x$ factor in Eq.~\eqref{eq:app-pS} is cancelled by the
$dz_-=x\,dy$ phase-space factor, while the product of the two positive
curvatures contributes $x^4$.  This proves the quartic scaling.

For completeness, the leading coefficient can also be written in
closed form.  Let
\begin{equation}
  z_0=\sqrt{\frac32}\,u
\end{equation}
and
\begin{align}
  I(z_0)
  ={}&
  \phi(z_0)(z_0^2+14)\Phi(z_0/2)
  +2z_0\phi(z_0)\phi(z_0/2)
  \nonumber\\
  &+
  25\sqrt{\frac{2}{5\pi}}\,
  \overline\Phi\!\left(
    \frac{\sqrt5}{2}z_0
  \right).
  \label{eq:app-I}
\end{align}
Then
\begin{equation}
  R_2(r;u)
  \sim
  \frac{A(u)}{\sigma^2}
  \left(\frac{r}{\sigma}\right)^4,
  \qquad
  A(u)
  =
  \frac{I(z_0)}
       {1728\sqrt6\,\pi^{3/2}}.
  \label{eq:app-A}
\end{equation}
Since $\lambda(u)$ has dimensions $1/\sigma$, the coefficient in
Eq.~\eqref{eq:g2-quartic} is
\begin{equation}
  \boxed{
  C(u)
  =
  \frac{A(u)}
       {[\sigma\lambda(u)]^2}.
  }
  \label{eq:app-C}
\end{equation}

\bibliographystyle{unsrturl}
\bibliography{references}

\end{document}